\documentclass[11pt]{article}
\usepackage[a4paper, margin=1.8cm]{geometry}
\usepackage{amsmath}
\usepackage{amssymb}
\usepackage{amsthm} 
\usepackage{stmaryrd}
\usepackage{tikz}
\usepackage{authblk}
\usepackage{scrextend}

\usepackage{esvect}

\newtheorem{theorem}{Theorem}[section]
\newtheorem{lemma}[theorem]{Lemma}
\newtheorem{proposition}[theorem]{Proposition}
\newtheorem{corollary}[theorem]{Corollary}

\theoremstyle{definition} 
\newtheorem{definition}[theorem]{Definition}

\theoremstyle{remark} 
\newtheorem{remark}[theorem]{Remark}

\theoremstyle{example}

\theoremstyle{example} 
\newtheorem{question}[theorem]{Question}

\newcommand{\lts}{ \mathcal{L}}

\newcommand{\lfp}{\operatorname{lfp}}
\newcommand{\gfp}{\operatorname{gfp}}

\newcommand{\sem}[1] {  \llbracket #1 \rrbracket  }  
\newcommand{\gsem}[1]{\llparenthesis{\,#1\,}\rrparenthesis}

\newcommand{\Dual}{\mathrm{negate}}

\newcommand{\LTL}{\ensuremath{\textbf{LTL}}}
\newcommand{\CTL}{\ensuremath{\textbf{CTL}}}
\newcommand{\CTLstar}{\ensuremath{\textbf{CTL}^*}}
\newcommand{\PCTL}{\ensuremath{\textbf{PCTL}}}
\newcommand{\PCTLstar}{\ensuremath{\textbf{PCTL}^*}}

\newcommand{\Lukmu}{\ensuremath{\text{\L}\mu}}
\newcommand{\qLmu}{\ensuremath{\text{qL}\mu}}

\newcommand{\strongand}{\odot}
\newcommand{\strongor}{\oplus}
\newcommand{\weakand}{\sqcap}
\newcommand{\weakor}{\sqcup}

\newcommand{\Mu}[2]{\mu #1.\, #2}
\newcommand{\Nu}[2]{\nu #1.\, #2}
\newcommand{\negate}[1]{\overline{#1}}

\newcommand{\GLE}[2]{#1 \vdash #2}

\newcommand{\Implies}{\rightarrow}
\newcommand{\dominates}{\vartriangleright}

\newcommand{\FVI}{\mathrm{FVI}}

\providecommand{\OO}[1]{\mathop{\mathrm{O}}\bigl(#1\bigr)}

\begin{document}

\title{{\Huge {\L}ukasiewicz $\mu$-calculus}}


\author[1]{Matteo Mio}
\author[2]{Alex Simpson}
\affil[1]{CNRS and \'{E}cole Normale Sup\'{e}rieure de Lyon, France}
\affil[2]{Faculty of Mathematics and Physics, University of Ljubljana, Slovenia}
\renewcommand\Authands{ $ \ \ \ \ $}

\date{}
%


\maketitle

\begin{abstract}
The paper explores properties of the \emph{{\L}ukasiewicz $\mu$-calculus}, or \Lukmu\ for short, an extension of {\L}ukasiewicz logic with scalar multiplication and least and greatest fixed-point operators (for monotone formulas). We observe that \Lukmu\  terms, with $n$ variables, define monotone piecewise linear functions from $[0,1]^n$ to $[0,1]$.  
Two effective procedures for calculating the output of \Lukmu\ terms on rational inputs are presented.
We then consider the {\L}ukasiewicz \emph{modal} $\mu$-calculus, which is obtained by adding box and diamond modalities to \Lukmu. Alternatively, it can be viewed as a generalization of Kozen's modal $\mu$-calculus adapted to probabilistic nondeterministic transition systems (PNTS's). We show how properties expressible in the well-known logic \PCTL\ can be encoded as {\L}ukasiewicz modal $\mu$-calculus formulas. We also show that the algorithms for computing values of {\L}ukasiewicz $\mu$-calculus terms provide automatic (albeit impractical) methods for verifying {\L}ukasiewicz modal $\mu$-calculus properties of finite rational PNTS's.
\end{abstract}


\section{Introduction}
\label{section:introduction}

This paper has two distinct motivations. The first is purely logical. The paper
investigates 
certain properties of  a very natural extension of {\L}ukasiewicz many-valued logic, obtained by adding scalar multiplication and least and greatest fixed-point operators for monotone formulas. Whereas, an extension with scalar multiplication has previously been investigated (for example in \cite{RMV2011}), the addition of monotone least and greatest fixed-points  has not been systematically studied before. (See Section~\ref{section_lmu}, however, for discussion of related work by 
Spada~\cite{Spada2008}.)

Our second motivation is more practical. We propose {\L}ukasiewicz logic with fixed points as a useful logic for the specification and verification of systems combining nondeterminism and probabilistic behaviour. We now summarize the basis of this perspective.

Among logics for expressing properties of nondeterministic (including concurrent) processes, represented as transition systems, Kozen's modal $\mu$-calculus~\cite{Kozen83} plays a fundamental r{o}le. It subsumes other temporal logics of processes, such as $\LTL$, $\CTL$ and $\CTLstar$. It does not distinguish bisimilar processes, but separates (finite) non-bisimilar ones. 
More generally, 
by a remarkable result of Janin and Walukiewicz~\cite{JW96},
it is exactly as expressive as 
the bisimulation-invariant fragment of monadic second-order logic. Furthermore, there is an intimate connection with parity games, which offers an intuitive reading of fixed-points, and underpins the existing technology for  model-checking $\mu$-calculus properties. 

For many purposes, it is useful to add probability to the computational model, leading to probabilistic nondeterministic transition systems, cf.~\cite{S95}.
Among the different approaches that have been followed to 
developing analogues of the modal $\mu$-calculus in this setting, the most significant is that introduced independently by Huth and  Kwiatkowska~\cite{HM96} and
by Morgan and McIver~\cite{MM97}, under which a \emph{quantitative} interpretation is given, with formulas denoting values in $[0,1]$.
This quantitative setting permits several variations. In particular, three different quantitative extensions of conjunction from booleans to $[0,1]$ (with  $0$ as false and $1$ as true) arise naturally~\cite{HM96}: 
minimum, $\min(x,y)$; multiplication, $xy$; and the 
strong conjunction (a.k.a.\ {\L}ukasiewicz t-norm)
from {\L}ukasiewicz fuzzy logic, $\max( x+y-1, \, 0)$. 
In each case, there is a dual operator 
giving a corresponding extension of disjunction: maximum, $\max(x,y)$; comultiplication, 
$x+y-xy$; and {\L}ukasiewicz strong disjunction,  $\min(x+y,\, 1)$. The choice of $\min$ and $\max$ for conjunction and disjunction  is particularly natural, since the corresponding quantitative modal $\mu$-calculus, called $\qLmu$ in \cite{MM07},
has an interpretation in terms of two-player \emph{stochastic} parity games, which extends the usual parity-game interpretation of the ordinary modal $\mu$-calculus. 
This allows the real number denoted by a formula to be understood as the \emph{value} of the associated game~\cite{MM07,MIO2012a}.

The present paper contributes to a programme of ongoing research, one of whose
overall aims is to investigate the extent to which quantitative modal $\mu$-calculi
play as fundamental a r{o}le in the probabilistic setting as that
of Kozen's modal $\mu$-calculus in the nondeterministic setting. 
The logic $\qLmu$, with  min/max as conjunction/disjunction, is insufficiently expressive. For example, it cannot 
encode the standard probabilistic
temporal logic  \PCTL\ of~\cite{BA1995}.
Nevertheless, richer calculi can be obtained by augmenting $\qLmu$ with the 
other alternatives for conjunction/disjunction, to be used in combination with $\max$ and $\min$. Such extensions  were investigated by the first author in~\cite{MIO2012b,MioThesis}, where 
the game-theoretic interpretation was generalized to accommodate the new operations.

In this paper, we study the power of least and greatest
fixed points in the presence of both the weak
conjunction and disjunction ($\min$ and $\max$---written as 
$\weakand$ and $\weakor$) as well as  the {\L}ukasiewicz strong connectives
(written as $\strongand$ and $\strongor$) in combination. In addition, as
already mentioned we include a basic operation for
multiplying the value of a formula by a rational constant in $[0,1]$.
We call the resulting
 logic \emph{{\L}ukasiewicz  
$\mu$-calculus}. 
We condsider both a basic version with just the operations mentioned above, and a modal version with added $\Box$ and $\Diamond$ modalities. 


As a first contribution of this paper, we consider the problem of
computing the value of a (non-modal) {\L}ukasiewicz $\mu$-calculus formula. 
We describe two approaches to this. The first is indirectly obtained by a reduction to the decision problem of the first order theory of rational linear arithmetic, which admits quantifier elimination. As a consequence of this property, {\L}ukasiewicz $\mu$-calculus terms always denote piecewise linear (generally discontinuous) functions. The second algorithm adapts the 
approximation-based approach to nested fixed-point calculation to our quantitative setting.

As a second contribution, we show that the the well-known probabilistic temporal logic \PCTL~\cite{BA1995} can be encoded in the modal version of
{\L}ukasiewicz $\mu$-calculus. A similar translation was originally given in 
the first author's PhD thesis~\cite{MioThesis}.
In this journal presentation of the result, we streamline the treatment slightly by removing the connectives of \emph{independent product} and its dual from the target $\mu$-calculus.
The translation  provides some evidence that 
{\L}ukasiewicz modal $\mu$-calculus provides an expressive logic for specifying properties of systems combining probabilic choice  and nondeterminism. 

Finally, we show how the computation of the value of a {\L}ukasiewicz modal $\mu$-calculus formula in a finite probabilistic transition system can be reduced to the problem of computing the value of a formula in the modality-free 
calculus. This, in theory, allows the algorithms we consider for computing such values to be brought to bear on the quantitative model-checking problem for {\L}ukasiewicz modal $\mu$-calculus. Also, again in theory, it subsumes the problem of model-checking $\PCTL$. 
However, whereas the model-checking problem for $\PCTL$ is known to be polynomial-time solvable in the size of formula and model \cite[{\S}10.6.2]{BaierKatoenBook}, the algorithms considered in the present paper have abysmal complexity, at least on paper. To what extent this can be improved is an interesting topic for future work. 

In any case, the purpose of the present paper is not to provide a practical algorithmic framework. Rather, it is to explore something of the richness of 
the combination of {\L}ukasiewicz logic with fixed points, with the idea that this provides a conceptually clean basis for the specification of probabilistic systems, and one that is worthy of further investigation as a potential tool 
for verification. 

\section{{\L}ukasiewicz $\mu$-calculus}\label{section_lmu}

This paper concerns {\L}ukasiewicz many-valued (``fuzzy'') logic, see, e.g.,  \cite{HJ98,MundiciBook}. The main purpose is to study fixed-point extensions of {\L}ukasiewicz logic. As discussed in Section~\ref{section:introduction}, our motivation for considering such extensions comes from an interest in specification logics for probabilistic systems. For this reason, we consider only the standard interpretation of the connectives in the interval $[0,1]$, rather than, say, in arbitrary MV-algebras \cite{MundiciBook}. In this context, it is natural also to extend formulas with a construct for ``scalar multiplication'' by a real $r$. (Such an extension with scalar multiplication is studied in \cite{RMV2011}, see also \cite{Gerla2001,Gerla2001b}, where it is shown to be the logical counterpart of so-called Riesz MV-algebras.)

This leads to a syntax of formulas 
generated by the following grammar (we do not try to be minimal),
\[
t \, ::= \, 
x \mid \underline{0}  \ \mid \underline{1} \mid 
r \, t  \mid t \weakor t \mid t \weakand t \mid t \strongor t \mid t \strongand t \ | \ \neg t \enspace ,
\]
where $x$ ranges over a countably infinite set of variables, and  $r$ ranges over reals in $[0,1]$.  Henceforth, we call the $t$ specified by the grammar above \emph{terms} rather than formulas.

We write $t(x_1, \dots, x_n)$ to mean that all variables of $t$ are contained in $\{x_1, \dots, x_n\}$. Such a term is interpreted as  a function of type $[0,1]^n\rightarrow[0,1]$, by inductively defining the \emph{value} $t(\vv{r})$ of $t(x_1, \dots, x_n)$ applied to a vector 
$(r_1, \dots, r_n) \!\in\! [0,1]^n$ as follows:
\begin{align*}
x_i(\vv{r}) & = r_i \ \ \ \ \ \ \   \underline{0}(\vv{r})  = 0 \ \ \ \ \ \ \ \underline{1}(\vv{r})  = 1& \textnormal{variables and constants}
\\
(q\, t)(\vv{r}) &  = q \cdot  t({\vv{r}}) & \textnormal{scalar multiplication}
\\
(t_1 \weakor t_2)(\vv{r}) & = \max(t_1(\vv{r}), t_2(\vv{r})) & \textnormal{{\L}ukasiewicz weak disjunction}
\\
(t_1 \weakand t_2)(\vv{r}) & = \min(t_1(\vv{r}), t_2(\vv{r}))& \textnormal{{\L}ukasiewicz weak conjunction}
\\
(t_1 \strongor t_2)(\vv{r}) & = \min(t_1(\vv{r})+ t_2(\vv{r}),\, 1) & \textnormal{{\L}ukasiewicz strong disjunction}
\\
(t_1 \strongand t_2)(\vv{r}) & = \max(t_1(\vv{r})+ t_2(\vv{r}) - 1, \, 0)& \textnormal{{\L}ukasiewicz strong conjunction}
\\
(\neg t)(\vv{r}) & = 1 - t(\vv{r})& \textnormal{negation}
\end{align*}

One way of extending \L ukasiewicz logic with least and greatest fixed points has previously been considered in  \cite{Spada2008}. This is based on the observation that, since any term $t(x_1,\dots,x_n,y)$ denotes a continuous function, the existence of a solution $y$ to the equation 
$
y  =  t\big(x_1\dots x_n,y\big)
$
is guaranteed by the Brouwer fixed-point theorem. Because the space of all solutions is closed, 
it is thus possible to define the semantics of a least-fixed-point expression $\mu y. t(x_1\dots x_n,y)$, with $ t(x_1\dots x_n,y)$ an arbitrary {\L}ukasiewicz term, as the smallest solution; and similarly for greatest-fixed-point expressions $\mu y. t(x_1\dots x_n,y)$. For example the term $\mu x. \neg x$ denotes, under this approach, the rational number $\frac{1}{2}$ (as does $\nu x. \neg x$). However, solutions of fixed-point expressions need not be themselves continuous. For example, $\mu y. (x \oplus y)$ denotes (see Proposition \ref{semantics_threshold1} below) the discontinuous function
\begin{equation}
\label{equation:discontinuous}
f(x) \; = \;   \begin{cases}
 						1 & \text{if $ x > 0$}  \\
						0 & \text{otherwise} \\
						\end{cases}      \enspace .
\end{equation}
As a consequence, the term $\mu x. \neg \big( \mu y. (x \oplus y )\big)$ cannot be given meaning because  the equation $x = 1- f(x)$ does not have solutions. For this reason, the fixed-point extension of \L ukasiewicz logic of \cite{Spada2008} does not admit nesting of fixed-point operators.

In this paper we follow an alternative approach to the assignment of values to fixed-point expressions, which allows unrestricted nesting of fixed-point operators.
For this, we base our treament of fixed points on the 
order-theoretic Knaster-Tarski fixed-point theorem, rather than on the topological Brouwer fixed-point theorem. This approach has been very productive in the development of $\mu$-calculi, see, e.g.,~\cite{Kozen83,Stirling96}. As we shall show, it also leads to a rich and interesting theory
in the context of \L ukasiewicz logic.

As is necessary for the application of the Knaster-Tarski fixed-point theorem, 
we consider fixed points of monotone functions only. This is implemented syntactically by considering only terms without negation. 
The syntax of \emph{{\L}ukasiewicz $\mu$-terms} (or just \emph{$\mu$-terms} for short) is thus specified by the grammar:
\[
t \, ::= \, 
x \mid \underline{0}  \ \mid \underline{1} \mid 
r \, t  \mid t \weakor t \mid t \weakand t \mid t \strongor t \mid t \strongand t  \mid \Mu{x}{t} \mid \Nu{x}{t}
\]
where, as expected, the $\mu$ and $\nu$ operators bind their variables. We write $t(x_1\dots x_n)$ to specify that the \emph{free variables} occurring in $t$ are contained in $\{x_1,\dots, x_n\}$. 
We say that a $\mu$-term $t$ is \emph{rational} if all scalars $r$ appearing in scalar multiplications in $t$  are rational numbers.

The \emph{value} $t(\vv{r})$ of a  $\mu$-term $t(x_1, \dots, x_n)$ applied to a vector 
$(r_1, \dots, r_n) \!\in\! [0,1]^n$ is specified by extending the definition of value of {\L}ukasiewicz terms as follows: 
\begin{align*}
\big(\Mu{y}{t(x_1\dots x_n,y})\big)(\vv{r}) & = \lfp (r^\prime \mapsto  t(\vv{r},r^\prime))
\\
\big(\Nu{y}{t(x_1\dots x_n,y})\big)(\vv{r}) & = \gfp (r^\prime \mapsto  t(\vv{r},r^\prime))
\end{align*}
\noindent
where $\lfp$ and $\gfp$ are the operators returning the least and greatest fixed point respectively. These fixed points exist, because $t(x_1, \dots, x_n)$ always defines a monotone function from $[0,1]^n$ to $[0,1]$ (by routine induction on $t$). Note that the values of $\mu x.x$ and $\nu x.x$ are $0$ and $1$ respectively. Thus the constants $\underline{0}$ and $\underline{1}$ of {\L}ukasiewicz logic are derivable  as $\mu$-terms. Given a real $r\!\in\![0,1]$, we write $\underline{r}$ for the formula $r\, \underline{1}$ whose value is $r$.

As is customary in fixed-point logics, we have presented terms in positive form, i.e., without negations. However, an operation $\Dual(t)$ can be defined on 
 terms by replacing every connective with its dual, replacing $(q\, t)$ with $(q\, \Dual(t)) \oplus \underline{1-q} $,
and using the definitions $\Dual(\Mu{x}{t}) = \Nu{x}{\Dual(t)}$ and 
$\Dual(\Nu{x}{t}) = \Mu{x}{\Dual(t)}$ for fixed points. 
It is simple to verify that, for a term $t(x_1, \dots, x_n)$, it holds that
$\Dual(t)(\vv{r})= 1-t(1-\vv{r})$ , for all $\vv{r} \in [0,1]^n$.
This defines negation for \emph{closed} terms $t$, since for such terms the equality 
$\Dual(t) = 1-t$ holds between values.
 
 The next definition introduces a few useful macro formulas.
\begin{definition}\label{thresholds_encoding}
For every $\mu$-term $t$ define:
\begin{center}
\begin{tabular}{l l l l}
$\mathbb{P}_{>0}t= \mu y. (y \oplus t)$ &  
$\mathbb{P}_{=1}t= \nu y. (y \odot t) $&
$ \mathbb{P}_{> r}t=   \mathbb{P}_{>0}(t \odot \underline{1-r})$&
$\mathbb{P}_{\geq r}t=   \mathbb{P}_{=1}(     t \oplus \underline{1-r})$ \\
\end{tabular}
\end{center}
where $r\!\in\! (0,1)$ and the variable $y$ does not appear in $t$.  We write $\mathbb{P}_{\rtimes r}t$, for $r\!\in\! [0,1]$, to denote one of the four cases. The derived operators $\mathbb{P}_{\rtimes r}$ are called \emph{threshold modalities}.
\end{definition}

\begin{proposition}\label{semantics_threshold1}
Given a $\mu$-term $t(x_1\dots x_n)$ the following equality holds:
\begin{center}
$(\mathbb{P}_{\rtimes r}t)(\vv{r})= \left\{     \begin{array}{l  l}
 						1 & $if $t(\vv{r}) \rtimes r \\
						0 & $otherwise$\\
						\end{array}      \right.$
\end{center}
\end{proposition}

\noindent
Therefore $\mu$-terms generally represent discontinuous functions. 

The main goal of this section is to establish the following fundamental properties of the {\L}ukasiewicz $\mu$-calculus. First, any 
$\mu$-term $t(x_1, \dots, x_n)$ defines a piecewise linear function 
in $[0,1]^n \to [0,1]$. Second, if the term is rational, then so are the linear pieces; by which we mean that the coefficients and constants in the linear expression for each piece are rational, as are those in the linear expressions specifying the constraints defining the domain of the piece. In particular, if $t$ is a closed rational $\mu$-term, then its value is a rational number. Finally, we show that this rational number is computable from $t$.

In this section, we formalise and prove these results via
a simple reduction to the first-order theory of linear arithmetic.
After this, in Section~\ref{section:algorithm}, we provide an alternative
approach based on a direct iterative algorithm for computing least and greatest fixed points of monotone piecewise linear functions.

A \emph{linear expression} in variables $x_1, \dots, x_n$ is
an expression 
\[
q_1 x_1 + \dots + q_n x_n + q
\]
where $q_1, \dots, q_n, q$ are real numbers. We say that
the linear expression is \emph{rational} when all of $q_1, \dots, q_n, q$ are
rational numbers. (The choice of the letter $q$ reflects on the fact that our primary interest will be in the rational case.)
We write $e(x_1, \dots, x_n)$ if $e$ is a linear expression in $x_1, \dots, x_n$, in which case, given real numbers $r_1, \dots, r_n$, we write $e(\vv{r})$ for the value of the expression when the variables $\vv{x}$ take values $\vv{r}$.
We also make use of the closure of linear expressions under substitution: given $e(x_1, \dots, x_n)$ and $e_1(y_1, \dots, y_m), \dots, e_n (y_1, \dots, y_m)$,
we write $e(e_1, \dots, e_n)$ for the evident substituted expression in variables $y_1,\dots, y_m$ (which is defined formally by multiplying out and adding coefficients).

The first-order 
theory of \emph{linear arithmetic} has 
linear expressions as terms, and strict and non-strict 
inequalities between linear expressions,
\begin{equation}
\label{eqn:inequalities}
e_1 < e_1 \qquad e_1 \leq e_2 \enspace ,
\end{equation}
as atomic formulas. 
Equality can be expressed as the conjunction of two non-strict inequalities and the negation of an atomic formula can itself be expressed as an atomic formula. 
The truth of a first-order formula is given via its interpretation in the reals.
The theory of \emph{rational} linear arithmetic is defined identically, but with terms restricted to rational linear expressions.

Both linear arithmetic and rational linear arithmetic enjoy quantifier elimination; see, e.g., \cite{Ferrante1975}. As a consequence, they are model complete. Thus, in the case of rational linear arithmetic, the rational numbers form a model, and the inclusion of the rationals into the reals is an elementary embedding.

\begin{proposition}
\label{proposition:mu-term:arithmetic}
For every {\L}ukasiewicz $\mu$-term $t(x_1, \dots, x_n)$, 
its graph  
\[
\{(\vv{x},y) \in [0,1]^{n+1} \mid t(\vv{x}) = y\}
\] is definable by a formula $F_t(x_1, \dots, x_n, y)$ in the first-order theory of linear arithmetic. In the case that $t$ is rational, it holds that  $F_t$ is a rational formula and $F_t$ is computable from $t$. Furthermore,  if $u$ is the length of the $\mu$-term $t$ and $v$ is the number of fixed points constructors in $t$, the
length of $F_t$ is bounded by $2^v u c$, for some constant $c$.
\end{proposition}

\begin{proof}
The proof is a straightforward induction on the structure of $t$. We consider two cases, in order to illustrate the simple  manipulations  used in the construction of $F_t$.

If $t$ is $t_1 \strongor t_2$ then $F_t(\vv{x},y)$ is the formula
\[
\exists z_1, z_2 .\: F_{t_1}(\vv{x},z_1) \, \wedge \,
  F_{t_2}(\vv{x},z_2) \, \wedge \,
  \left((z_1 + z_2 \leq 1 \wedge y = z_1 + z_2) \, \vee \, 
    (1\leq z_1 + z_2  \wedge y = 1)
  \right)
\]
\noindent
If $l_1,l_2$ are the lengths of the formulas $F_{t_1}, F_{t_2}$ respectively, then the length of $F_{t_1 \strongor t_2}$  is $l \: \leq \: l_1 + l_2 + c$.

If $t$ is $\Mu{x_{n+1}}{t'}$ then $F_t$ is the formula
\[
F_{t'}(x_1, \dots, x_n, y,y) \, \wedge \, \forall z.\: F_{t'}(x_1, \dots, x_n, z,z) \Implies y \leq z \enspace .
\]
If $l'$ is the length of $F_{t'}$ then the 
length of $F_{\Mu{x_{n+1}}{t'}}$ is $l \: \leq \: 2l' + c$.

\end{proof}
\noindent

Proposition~\ref{proposition:mu-term:arithmetic} provides the following method of computing the value $t(\vv{q})$ of  rational \L ukasiewicz $\mu$-calculus term $t(x_1, \dots, x_n)$ at a rational vector $(q_1, \dots, q_n) \in [0,1]^n$. First construct $F_t(x_1, \dots, x_n,y)$. Next, perform quantifier elimination to obtain an equivalent 
quantifier-free formula $G_t(x_1, \dots, x_n,y)$, and consider its 
instantiation $G_t(q_1, \dots, q_n,y)$ at $\vv{q}$. (Alternatively, obtain 
an equivalent formula $G^{\vv{q}}_t(y)$ by performing quantifier elimination
on $F_t(q_1, \dots, q_n,y)$.) By performing obvious simplifications of atomic formulas in one variable, $G_t(q_1, \dots, q_n,y)$ reduces to a boolean combination of inequalities each having one of the following forms
\[
y \leq q \qquad y < q \qquad y \geq q \qquad y > q \enspace .
\]
By the correctness of $G_t$ there must be a unique rational satisfying the Boolean
combination of constraints, and this can be extracted in a straightforward way from $G_t(q_1, \dots, q_n,y)$. The quantifier-elimination procedure of~\cite{Ferrante1975}, when given a formula of length $l$ as input produces a formula of length at most $2^{dl}$ as output and takes time at most $2^{2^{d'l}}$,  for some constants $d$ and $d'$.
Thus the length of the formula $G_t(x_1, \dots, x_n,y)$ is
bounded by 
$\OO{2^{2^v u}}$, 
and the computation time for $t(\vv{q})$  is $\OO{2^{2^{2^v u}}}$, using a unit cost model for rational arithmetic. 

Better bounds can be obtained using the decision procedure for linear arithmetic of \cite{BJW2005} which is based on automata theoretic methods. Given the formula $F_t(q_1, \dots, q_n,y)$ of length $l$, one can construct in time $2^{dl}$, for some constant $d$, a deterministic B\"{u}chi automaton $\mathcal{A}$ recognizing the set of ($\omega$-words which are codings of) real numbers $r$ satisfying  $F_t(q_1, \dots, q_n,r)$. Once again, by the correctness of $F_t$, a unique rational number $q$ satisfies the formula $F_t(q_1, \dots, q_n,y)$. The code of $q$ can efficiently be extracted from $\mathcal{A}$. Thus this procedure computes the value of $t(\vv{q})$ in time bounded by $\OO{2^{2^{v}u}}$.

\begin{remark}
The practical efficiency of the procedure can be improved by a more careful translation from $\mu$-terms to linear arithmetic (for example, solving adjacent fixed points of the same kind simultaneously to avoid unnecessary  quantifier alternation). Nevertheless, even with such improvements, the theoretical bounds discussed above are not reduced. It might be possible to improve the doubly exponential bound to a single exponential one by developing an efficient direct translation of $\mu$-terms into the corresponding B\"{u}chi automaton $\mathcal{A}$ (e.g., exploiting the fact that the formula $F_{\mu x.t^\prime}$ contains two identical occurrences of $F_{t^\prime}$) thus avoiding the exponential blow-up required by the intermediate translation into the formula $F_t(q_1, \dots, q_n,y)$. We leave this question open for further research.

\end{remark}

We now elaborate the sense  in which the function $t(x_1, \dots, x_n)$ is piecewise linear. 
A \emph{conditioned linear expression} is  a pair, written 
$\GLE{C}{e}$,
where $e$ is a linear expression, and $C$ is a finite set of strict and non-strict inequalities between linear expressions; i.e., each element of
$C$ has one of the forms in~(\ref{eqn:inequalities}).
We write $C(\vv{r})$ for the conjunction of the inequations obtained by instantiating $\vv{r}$ for 
$\vv{x}$ in $C$. 
The role of a conditioned linear expression, $\GLE{C}{e}$, is to
specify one piece of a piecewise linear function. The domain of the piece is the set of vectors $\vv{r}\,$  for which $C(\vv{r})$ is true. And the expression $e$ specifies the linear function that applies over that domain. 
Note that the domain $\{(r_1, \dots, r_n) \mid C(\vv{r})\}$ is always convex; i.e., if
$C(\vv{r})$ and $C(\vv{s})$ then, for all $\lambda \in [0,1]$, we have
$C(\lambda\vv{r} + (1-\lambda)\vv{s})$.
This fact will be exploited in the sequel. We remark that the domain need not be  open or closed.
It may also have empty interior. 

Let $\mathcal{F}$ be a \emph{system} (i.e., finite set) of conditioned linear expresssions in variables $x_1, \dots, x_n$. We say that $\mathcal{F}$ \emph{represents} a function $f \colon [0,1]^n \to [0,1]$ if the following conditions hold:
\begin{enumerate}
\item For all $r_1, \dots, r_n \in [0,1]$, there exists a conditioned linear expression $(\GLE{C}{e}) \in \mathcal{F}$ such that $C(\vv{r})$ is true, and

\item for all $r_1, \dots, r_n \in [0,1]$, and every conditioned linear expression $(\GLE{C}{e}) \in \mathcal{F}$, if  $C(\vv{r})$ is true then $e(\vv{r}) = f(\vv{r})$.

\end{enumerate}

\noindent
Note that, for two conditioned linear expressions $(\GLE{C_1}{e_1}), (\GLE{C_2}{e_2}) \in \mathcal{F}$, we do not require different conditioning sets $C_1$ and $C_2$ to be disjoint. However, $e_1$ and $e_2$ must agree on any overlap.

Obviously, the function represented by a system of conditioned linear expressions is unique, when it exists. But not every system represents a function. In general, one could impose syntactic conditions on a system to ensure that it represents a function, but we shall not pursue this. 


\begin{definition}
We say that a function $f\colon [0,1]^n \to [0,1]$ is \emph{piecewise linear} if there exists a {system} $\mathcal{F}$ of conditioned linear expresssions in variables $x_1, \dots, x_n$ that represents $f$. We say that $f$ is \emph{rational} piecewise linear if there exists such an $\mathcal{F}$ containing only rational numbers.
\end{definition}
We emphasise that piecewise linear functions, as defined above, need not be continuous (see, e.g., (\ref{equation:discontinuous}) above and 
Section~\ref{section:motivating} for examples). 
The result below, which is presumably folklore, justifies the definition we have given.

\begin{proposition}
\label{proposition:cle:arithmetic}
A function $f\colon [0,1]^n \to [0,1]$ is piecewise linear 
 if and only if its graph 
$\{(\vv{x},y) \in [0,1]^{n+1} \mid f(\vv{x}) = y\}$ is definable by a formula $F(x_1, \dots, x_n, y)$ in the first-order theory of linear arithmetic. 
Similarly, it is rational piecewise linear if and only if its graph is definable by a formula of rational linear arithmetic. In the rational case, 
a defining formula and a representing system of conditioned linear equations can each be computed from the other.
\end{proposition}
\begin{proof}
The proof  is a straightforward application of quantifier elimination.
Suppose we have a system of $k$ conditioned linear expressions representing $f$.
Each conditioned expression $\GLE{C\,}{\,e}$ is captured by the implication $(\bigwedge C) \rightarrow y = e$, so the whole system translates into a conjunction of $k$ such implications. To this conjunction, one need only add the range constraints $0 \leq 
z$ and $z \leq 1$ for each variable $z$, as further conjuncts. In this way, the graph is easily expressed as a quantifier free formula. (Since the implications are equivalent to disjunctions of atomic formulas, the resulting formula is naturally in conjunctive normal form.) 

Conversely, suppose $F(x_1, \dots, x_n, y)$ defines the graph of $f$. By  quantifier elimination, we can assume that $F$ is quantifier free and in disjunctive normal form. Then $F$ is a disjunction of conjunctions, where each conjunction, $K$, can be easily rewritten in the form
\begin{equation}
\label{equation:K}
\left(\bigwedge C  \right)
\, \wedge \, 
\left(\bigwedge_{1 \leq i \leq h} y > a_i \right)
\, \wedge \, 
\left(\bigwedge_{1 \leq i \leq k} y \geq b_i\right)
\, \wedge \, 
\left(\bigwedge_{1 \leq i \leq l} y \leq c_i\right)
\, \wedge \, 
\left(\bigwedge_{1 \leq i \leq m} y < d_i\right) \enspace ,
\end{equation}
such that the only variables in the finite set of atomic formulas $C$, and linear expressions $a_i, b_i, c_i, d_i$ are $x_1, \dots, x_n$. 
Since $F$ is the graph of a function, for all reals
$r_1, \dots, r_n$, there is at most one $s$ such that 
$K(\vv{r},s)$ holds, and, if it does, then all of $r_1, \dots, r_n, s$ are in $[0,1]$.
Given such an $s$, we therefore  have:
\[
\max\{a_i(\vv{r}) \mid 1 \leq i \leq h\} 
< 
\max\{b_i(\vv{r}) \mid 1 \leq i \leq k\} = s = \min\{c_i(\vv{r}) \mid 1 \leq i \leq l\} < \min\{d_i(\vv{r}) \mid 1 \leq i \leq m\} \enspace .
\]
A system of conditioned linear expressions for $f$ is thus obtained as follows.
For each conjunct $K$ in $F$, written in the form of (\ref{equation:K}) above,
and each $j$ with $1 \leq j \leq k$, include the conditioned linear expression:
\[
\GLE{
C, \,  
\{b_j > a_i \}_{1 \leq i \leq h}, \, 
\{b_j \geq b_i\}_{1 \leq i \leq k}, \, 
 \{b_j \leq c_i\}_{1 \leq i \leq l}, \,
 \{b_j <  d_i\}_{1 \leq i \leq m}, \,
}{\, b_j} \enspace .
\]
\end{proof}

\noindent

Combining Propositions~\ref{proposition:mu-term:arithmetic} and~\ref{proposition:cle:arithmetic} we obtain:

\begin{corollary}
\label{corollary:sledgehammer}
For every {\L}ukasiewicz $\mu$-term $t(x_1, \dots, x_n)$, the function 
\[
\vv{r} \mapsto t(\vv{r}) \colon [0,1]^n \to [0,1]
\]
is piecewise linear. Moreover, if $t$ is a rational $\mu$-term then the function is rational piecewise linear, and a representing system of conditioned linear expressions in variables $x_1, \dots, x_n$ can be computed from $t$. 
\end{corollary}

McNaughton's Theorem  (\cite{McNaughton1951}, see also \cite{MundiciBook}) famously classifies the functions defined by ordinary {\L}ukasiewicz formulas (without scalar multiplication) as the continuous piecewise linear functions on $[0,1]$ with integer coefficients (McNaughton functions). The result is specialized in \cite[Thm. 3.5]{CK2013} where it is shown that monotone McNaughton functions correspond to monotone (i.e., negation free) ordinary {\L}ukasiewicz formulas. Another variant of McNaughton's theorem, proved in \cite[Thm. 4.5]{Gerla2001b}, classifies rational {\L}ukasiewicz formulas (i.e., with scalar multiplication by rationals) as the continuous piecewise linear functions on $[0,1]$ with rational coefficients.

A positive answer to the question below would provide an analogous result for the {\L}ukasiewicz $\mu$-calculus. 

\begin{question}
Is every monotone (rational) piecewise linear function $f \colon [0,1]^n \to [0,1]$ definable by a (rational) $\mu$-term $t(x_1, \dots, x_n)$?
\end{question}

\section{An iterative algorithm for evaluating $\mu$-terms}
\label{section:algorithm}

The computation of a representing system of conditioned linear expressions
for a $\mu$-term $t$ via quantifier elimination,
provided by the proofs of
Propositions~\ref{proposition:mu-term:arithmetic} and~\ref{proposition:cle:arithmetic}, is indirect.
In this section we  present an alternative algorithm for calculating
the value $t(\vv{r})$ of a rational $\mu$-term at rationals $r_1, \dots, r_n  \in [0,1]$, which is
directly based on manipulating conditioned linear expressions. 
Rather than computing an entire system of conditioned linear expressions representing $t$, the algorithm works locally to provide a single conditioned expression that applies to the input vector $\vv{r}$. Fixed points are computed by iterating through approximations to them, starting with $0$ for least fixed points and $1$ for greatest.

\subsection{Motivating the algorithm}
\label{section:motivating}

The algorithm works inductively on the structure of terms. The crucial aspect is how to compute least and greatest fixed points. Accordingly, given a method for obtaining linear pieces for an inner term $t'(x_1, \dots, x_{n+1})$, we need to specify how to find linear pieces for the fixed-point terms    $\Mu{x_{n+1}}{t'}$ and $\Nu{x_{n+1}}{t'}$. 

We illustrate the main idea behind the iteration by considering an example of finding the least-fixed point of the piecewise linear function $f$ specified by:
\begin{align*}
\GLE{x < \frac{1}{2} ~ &}{~\frac{1}{2}\, x + \frac{1}{4}}
\\
\GLE{\frac{1}{2} \leq x \leq \frac{9}{16} ~ &}{ ~ x + \frac{1}{8}}
\\
\GLE{\frac{9}{16} < x  \leq \frac{5}{8} ~ &}{~ \frac{3}{4}}
\\
\GLE{\frac{5}{8} \leq  x ~ &}{~ \frac{1}{4}\, x + \frac{19}{32}}
\end{align*}
This function, which is illustrated in Figure~\ref{figure:example},  is discontinuous and has a unique fixed-point at $\frac{19}{24}$.

\begin{figure}
\begin{center}
 \begin{tikzpicture}
          \draw[ultra thin,color=gray] (0.0,0.0) grid (4.0,4.0);

          \draw[->] (0,0) node[below] {0} -- (4,0) node[below] {1};
          \draw[->] (0,0) node[left] {0} -- (0,4) node[left] {1};

          \draw [domain=0:1.98,thick] plot (\x,1+0.5*\x);
          \draw [thick] (2,2) circle (1pt);
          \filldraw  (2,2.5) circle (0.8pt);
          \draw [domain=2:2.25,thick] plot (\x,0.5+\x);
          \filldraw  (2.25,2.75) circle (0.8pt);
          \draw [domain=2.27:2.5,thick] plot (\x,3);
          \draw [thick] (2.25,3) circle (1pt);
          \filldraw (2.5,3) circle (0.8pt);
          \draw [domain=2.5:4,thick] plot (\x,2.375+0.25*\x);

    \end{tikzpicture}
\end{center}
\caption{Example discontinuous piecewise linear function}
\label{figure:example}
\end{figure}
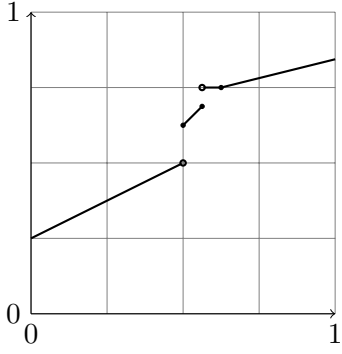

Our algorithm calculates the least fixed-point by iteratively refining
a lower-approximation $d$, starting from $d\!=\!0$. 
We note, however, that a blind iteration of $f$ from $0$ does not work.
The successive values $0\leq f(0)\leq f^2(0)\leq f^3(0) \leq \dots$ (i.e.,
$0 \leq \frac{1}{4} \leq  \frac{3}{8} \leq \frac{7}{16} \leq \dots$) never reach a fixed point. Furthermore, their limit $\frac{1}{2}$ is not a fixed point of $f$, due to the discontinuity of $f$ at $\frac{1}{2}$.

Instead the algorithm works by iterating through the linear pieces. The initial approximation is $d = 0$, and the linear piece for $x = 0$ is retrieved, given by $\frac{1}{2}\, x + \frac{1}{4}$
on the domain $[0,\frac{1}{2})$. The unique solution of 
$x = \frac{1}{2}\, x + \frac{1}{4}$ is $x = \frac{1}{2}$, but this lies outside the domain. So $d$ is replaced by a new approximation, given by
$\frac{1}{2}\, x + \frac{1}{4}$ calculated at the upper bound $x = \frac{1}{2}$ of the domain. That is, the next approximation is $d = \frac{1}{2}$.

We now again retrieve the linear piece at $x = \frac{1}{2}$, which is 
$x + \frac{1}{8}$ on the domain $[\frac{1}{2}, \frac{9}{16}]$. The equation
$x = x + \frac{1}{8}$ has no solution, so $d$ is replaced by the new approximation $x + \frac{1}{8}$ calculated at the upper bound $\frac{9}{16}$, i.e., $d = \frac{11}{16}$.

The linear piece at $x = \frac{11}{16}$ is $\frac{1}{4} \, x + \frac{19}{32}$ on the domain $[\frac{5}{8},1]$ (note that the algorithm has skipped one of the linear pieces of $f$).
The unique solution of $x = \frac{1}{4} \, x + \frac{19}{32}$ is $x = \frac{19}{24}$. Because this lies in the domain, 
$\frac{19}{24}$ is the desired least fixed point of $f$.

The full algorithm, detailed in Section~\ref{subsection:algorithm} below, adapts the approach outlined above to the general computation of $\Mu{x_{n+1}}{t'}$ for $t'(x_1, \dots, x_n, x_{n+1})$. In the case that $n \geq 1$, significant further complications arise due to the need to  calculate the fixed point as a linear expression rather than just a number. 

\subsection{The algorithm}
\label{subsection:algorithm}

The algorithm takes, as input, a rational $\mu$-term $t(x_1, \dots, x_n)$ and a vector of rationals $(r_1, \dots, r_n) \in [0,1]^n$, and returns a conditioned linear expression $\GLE{C}{e}$, in variables $x_1, \dots, x_n$, with the following two properties.
\begin{itemize}
\item[(P1)] $C(\vv{r})$ is true.

\item[(P2)] For all $s_1, \dots, s_n \in \mathbb{R}$, if $C(\vv{s})$ is true then
$s_1, \dots, s_n \in [0,1]$ and $e(\vv{s}) =  t(\vv{s})$.
\end{itemize}
\noindent
It follows that $e(\vv{r}) =  t(\vv{r})$, so $e$ can indeed be used to compute the value $t(\vv{r})$.

For the purposes of the correctness proof in Section~\ref{section:correctness}, it is convenient to consider the running of the algorithm in the more general case that $r_1, \dots, r_n$ are arbitrary real numbers in $[0,1]$. This more general algorithm can be understood as an algorithm in the 
\emph{Real RAM} model of computation, also known as the \emph{BSS} model after the initials of its proposers~\cite{BSS1989}. When the input vector is rational, all real numbers encountered during execution of the algorithm are themselves rational, and so the general Real RAM algorithm specialises to  a 
\emph{bona fide} (Turing Machine) algorithm in this case. Moreover, even in the case of irrational inputs, all linear expressions constructed in the course of the algorithm are rational.

The algorithm works recursively on the structure of the term $t$. We present illustrative cases for terms
$t_1 \strongor t_2$ and $\Mu{x_{n+1}}{t'}$. The latter is the critical case. 
The algorithm for $\Nu{x_{n+1}}{t'}$ is an obvious dualization.

If $t$ is $t_1 \strongor t_2$ then recursively compute $\GLE{C_1\,}{\,e_1}$ and
 $\GLE{C_2\,}{\,e_2}$. If $e_1(\vv{r}) + e_2(\vv{r}) \leq 1$ then return
\[
\GLE{C_1, \, C_2, \, e_1 + e_2 \leq 1  \, }{\, e_1 + e_2} \enspace .
\]
Otherwise, return
\[
\GLE{C_1, \, C_2, \, e_1 + e_2 \geq 1  \, }{\, 1} \enspace .
\]

In the case that $t$ is $\Mu{x_{n+1}}{t'}$, enter the following 
loop starting with $D = \emptyset$ and $d =0$. 

\paragraph{Loop:}
At the entry of the loop we have a finite set $D$ of inequalities between linear expressions in $x_1, \dots, x_n$, and we have a linear expression $d(x_1, \dots, x_n)$.
The loop invariant that applies is:
\begin{itemize}
\item[(I1)] $D(\vv{r})$ is true; and
\item[(I2)]  for all $\vv{s} \in [0,1]^n$, if $D(\vv{s})$ then $d(\vv{s}) \leq (\Mu{x_{n+1}}{t'})(\vv{s})$.
\end{itemize}
 \noindent
We think of $D$ as constraints propagated from earlier iterations of the loop, and of $d$ as the current approximation to the least fixed point subject to the constraints.

Recursively compute $t'(x_1, \dots, x_{n+1})$ at $(\vv{r},d(\vv{r}))$ as $\GLE{C}{e}$, where $e$ has the form:
\begin{equation}
\label{equation:e-form}
q_1\, x_1 + \dots + q_n \, x_n + q_{n+1}\,  x_{n+1} + q \enspace .
\end{equation}

In the case that $q_{n+1} \neq 1$, 
define the linear expression:
\begin{equation}
\label{equation:f}
f \: := \: \frac{1}{1-q_{n+1}} \, \left( \, q_1 \, x_1 + \dots + q_n \, x_n  + q \, \right) \enspace .
\end{equation}
Test if $C(\vv{r},f(\vv{r}))$ is true. If it is, 
exit the loop and return:
\begin{equation}
\label{equation:first-result}
\GLE{D \, \cup \, C(x_1, \dots, x_n, d(x_1, \dots, x_n) )\, \cup \, C(x_1, \dots, x_n, f(x_1, \dots, x_n) )\,}{\, f} 
\end{equation}
as the result of the algorithm for $\Mu{x}{t'}$ at $\vv{r}$.
Otherwise, if 
$C(\vv{r},f(\vv{r}))$ is  false, define  $N(x_1, \dots, x_n)$ to be the negation of the inequality $e_1(x_1, \dots, x_n, f(x_1, \dots x_n)) \vartriangleleft e_2(x_1, \dots, x_n, f(x_1, \dots x_n))$ (using 
$\vartriangleleft$ to stand for either $<$ or $\leq$), where
$e_1(x_1, \dots, x_{n+1})  \vartriangleleft e_2(x_1, \dots, x_{n+1})$ is a chosen inequality in $C$
for which $e_1(\vv{r}, f(\vv{r})) \vartriangleleft e_2(\vv{r}, f(\vv{r}))$ is false, and
go to \textbf{find next approximation} below.

In the case that $q_{n+1} = 1$, test the 
equality
$q_1 \, r_1 + \dots + q_n \, r_n  + q  = 0$.
If true, exit the loop with result:
\begin{equation}
\label{equation:second-result}
\GLE{D  \, \cup \, C(x_1, \dots, x_n, d(x_1, \dots, x_n) ) \, \cup \, \{q_1 \, x_1 + \dots + q_n \, x_n  + q  =  0\} \,} {\, d} \enspace .
\end{equation}
If instead 
$q_1 \, r_1 + \dots + q_n \, r_n  + q  \neq 0$, choose $N(x_1, \dots, x_n)$ to be whichever of 
the inequalities
\[
q_1 \, x_1 + \dots + q_n \, x_n  + q  \; <  \; 0 \qquad
0 \; < \; q_1 \, x_1 + \dots + q_n \, x_n  + q  
\]
is true for $\vv{r}$, 
and 
proceed with
\textbf{find next approximation} below.

\paragraph{Find next approximation:}
Arrange the inequalities in $C$ so they have the following structure. 
\begin{equation}
\label{equation:C}
C' 
\, \cup \, 
\{ x_{n+1} > a_i \}_{1 \leq i \leq l'} 
\, \cup \, 
\{ x_{n+1} \geq a_i \}_{l' < i \leq l} 
\, \cup \, 
\{ x_{n+1} \leq b_i \}_{1 \leq i \leq m'} 
\, \cup \, 
\{ x_{n+1} < b_i \}_{m' < i \leq m} 
\end{equation}
such that the only variables in the inequalities $C'$, and linear expressions $a_i, b_i$ are $x_1, \dots, x_n$. 
Choose $j$ with $1 \leq j \leq m$ such that 
$b_j(\vv{r}) \leq b_i(\vv{r})$ for all $i$ with $1 \leq i \leq m$ (in the sequel we shall refer to $b_j$ as the \emph{infimum term}).
Then go back to \textbf{loop}, 
taking
\begin{equation}
\label{equation:next-iteration}
D \, \cup \, C(x_1, \dots, x_n, d(x_1, \dots, x_n))\,  \cup \, \{N(x_1, \dots, x_n)\}
\, \cup \, \{b_j \leq b_i \mid 1 \leq i \leq m\}
\qquad
e(\vv{x},b_j(\vv{x}))
\end{equation}
to replace $D$ and $d$ respectively.




\subsection{A simple example}
\label{section:example}
Consider the rational $\mu$-term $t=\mu x. (     \mathbb{P}_{\geq \frac{1}{2}}x \, \sqcup \, \frac{1}{2})$, where $\mathbb{P}_{\geq \frac{1}{2}}x$ is the macro formula as in Definition~\ref{thresholds_encoding}, that is   $\mathbb{P}_{\geq \frac{1}{2}}x= \mathbb{P}_{=1}(    x \oplus \frac{1}{2}) = \nu y. ( y \odot (x \oplus  \frac{1}{2}))$. Thus,

$$
t=\mu x. \Big(  \nu y. \big( y \odot (x \oplus  \frac{1}{2})\big)  \sqcup \frac{1}{2}\Big)
$$
Here, $t^\prime(x)=  \nu y. \big( y \odot (x \oplus  \frac{1}{2})\big)  \sqcup \frac{1}{2}$ is a discontinuous function, and the value of $t$ is $1$.

We omit giving a detailed simulation of the algorithm on the subexpression $t'(x)$ at $x=r$. The result it produces, however,  is
$\{Ê0\leq x < \frac{1}{2}\}\vdash \frac{1}{2}$ if ${r}<\frac{1}{2}$, and 
$\{Ê\frac{1}{2} \leq x \leq 1\}\vdash 1$ if $r\geq \frac{1}{2}$.

We run the algorithm on input $\mu x.t^\prime(x)$. Set $D=\emptyset$ and $d=0$. Calculating $t^\prime(x)$ at $x = 0$ we obtain $C\vdash e$ as $\{Ê0\leq x < \frac{1}{2}\}\vdash \frac{1}{2}$. The value of the coefficient of $x$ in $e$ is $0$. Thus, we need to calculate $f:=\frac{1}{1-0}(\frac{1}{2})=\frac{1}{2}$. The constraint $C(\frac{1}{2})$ does not hold. Thus we need to iterate the algorithm with $D=\emptyset$ and $d\!=\!\frac{1}{2}$. Calculating $t^\prime(x)$ at $x=\frac{1}{2}$ produces $C\vdash e$ as $\{Ê\frac{1}{2} \leq  x \leq 1\}\vdash 1$. Compute $f:=\frac{1}{1- 0}(1)=1$. Since $C(1)$ holds, the algorithm terminates with $\GLE{\emptyset}{1}$, as desired.

\subsection{A more complex example}

As a second example, we compute the value of the $\mu$-term $\nu x_0.\, \mu x_1.\, (\frac{5}{8} \oplus \frac{3}{8} x_0) \odot (\frac{1}{2} \sqcup (\frac{3}{8} \oplus \frac{1}{2} x_1))$ to describe in some detail the behavior of the algorithm on nested fixed-point expressions.

We illustrate the execution using subscripts such as $d_0, d_1$ 
as reminders of the current  depth in the nested recursions that the algorithm is at. So, e.g., $d_0$ is the current approximation within the loop evaluating the $\nu x_0.( \dots) $ expression, and $d_1$ is the approximation in the inner loop evaluating the $\mu x_1 .(\dots) $ subexpression. 

The algorithm starts computing the value of the outer greatest fixed-point expression starting from $D_{0}=\emptyset$ and approximating function $d_0 := 1$. The conditional linear expression $C_0\vdash e_0$ associated with the inner $\mu$-term $\mu x_1.\, (\frac{5}{8} \oplus \frac{3}{8} x_0) \odot (\frac{1}{2} \sqcup (\frac{3}{8} \oplus \frac{1}{2} x_1))$ at value $(x_0 = 1)$ needs to be calculated.


This is computed iteratively starting from the condition $D_1=\emptyset$ and $d_1\!=\!0$. 
Recursively computing $(\frac{5}{8} \oplus \frac{3}{8} x_0) \odot (\frac{1}{2} \sqcup (\frac{3}{8} \oplus \frac{1}{2} x_1))$ at values $(x_0 = 1, x_1 = 0)$ one obtains (details are omitted) the conditional linear expression  $C_1 \vdash e_1$ with $C_1\!=\!\{ 0\leq x_0 \leq 1, \, x_1 \leq \frac{1}{4}\}$ and $e_1\!=\! \frac{3}{8}x_0 + \frac{1}{8}$. Note that the coefficient $q_1$ of the variable $x_1$, see (\ref{equation:e-form}), is $0$. The function $f_1$, see (\ref{equation:f}), is then defined as $f_1(x_0) \!=\! \frac{3}{8}x_0 + \frac{1}{8}$. The condition $C_1(1,f_1(1))\!=\!C_1(1, \frac{1}{2})$ does not hold, since $\frac{1}{2}\not\leq \frac{1}{4}$. Thus, the inequality $N_1(x_0) := \, \frac{3}{8}x_0 + \frac{1}{8} > \frac{1}{4}$, which is the negation of $f_1(x_0)\leq \frac{1}{4}$, is calculated. Simplifying, we have  $N_1(x_0) \!:=\! \, x_0 > \frac{1}{3}$. The supremum term in $C_1$, dual to the 
infimum term in (\ref{equation:C}), 
is $b_j(x_0) := \, \frac{1}{4}$, using the notation of (\ref{equation:C}). So the loop is repeated, see (\ref{equation:next-iteration}), with $D_1 := \{ 0 \leq x_0 \leq 1, \, d_1 \leq \frac{1}{4} , \, x_0 > \frac{1}{3}\}\!=\!\{ \frac{1}{3}<x_0< \frac{1}{4}\}$ and $d_1 := \, \frac{3}{8} x_0 + \frac{1}{8}$.

 Computing $(\frac{5}{8} \oplus \frac{3}{8} x_0) \odot (\frac{1}{2} \sqcup (\frac{3}{8} \oplus \frac{1}{2} x_1))$ with values $(x_0 = 1, x_1 = d_1(1) = \frac{1}{2})$ one obtains the conditional linear expression  $C_1 \vdash e_1$ with $C_1\!=\!\{ 0\leq x_0 \leq 1, \, \frac{1}{4}\leq x_1 \leq 1\}$ and $e_1\!=\! \frac{3}{8}x_0 + \frac{1}{2}x_1$. As in the previous iteration, $q_1\!=\!\frac{1}{2}\!\neq\!0$, and the function $f_1(x_0)$ is then defined as $f_1(x_0)\!=\!\frac{3}{4}x_0$.
 The algorithm checks that $C_1(1,f_1(1))\!=\!C_1(1,\frac{3}{4})$ holds and terminates with result $C_0 \vdash e_0$ equals to $\frac{1}{3} < x_0 \leq 1, \, \frac{1}{4} \leq \frac{3}{8}x_0 + \frac{1}{8} \leq 1, \, \frac{1}{4} \leq \frac{3}{4}x_0 \leq 1 \, \vdash \, \frac{3}{4}x_0$.

Simplifying, $C_0 \vdash e_0$ equals $\frac{1}{3} < x_0 \leq 1 \, \vdash \, \frac{3}{4}x_0$.
Since $q_0$ (the coefficient of the $x_0$ variable) is $\frac{3}{4}$, the function $f_0:=0$ is defined. The condition $C_0(f_0())\!=\!C_0(0)$ is not satisfied. Calculating the negated condition $N_0()$ is redundant, since the index is $0$. The infimum term in $C_0$ is $b_j()\!:=\!\frac{1}{3}$,
again using the notation of (\ref{equation:C}). Therefore the loop is repeated with $D_0\!=\!\emptyset$ and $d_0:=\frac{3}{4}\cdot \frac{1}{3}$.
 
The conditional linear expression associated with  the $\mu$-term $\mu x_1.\, (\frac{5}{8} \oplus \frac{3}{8} x_0) \odot (\frac{1}{2} \sqcup (\frac{3}{8} \oplus \frac{1}{2} x_1))$ at value $(x_0 = \frac{1}{4})$ needs to be calculated. This time we omit the details of the computation which produces the result: $C_0\vdash e_0$ with 
$C_0\!=\!\{ 0\leq x_0\leq \frac{1}{3}\}$ and $e_0(x_0)\!=\! \frac{3}{8}x_0+\frac{1}{8}$.
 
 The coefficient $q_0 \!=\! \frac{3}{8}$ is not equal to $1$, thus the function  $f_0() \!:=\! \frac{1}{5}$ is defined. The condition $C_0(f_0()) \!=\! C_0(\frac{1}{5})$ holds. Thus the algorithm returns the conditional linear expression $C_0(f_0)\cup C_0(d_0)\vdash \frac{1}{5}$, i.e., $\emptyset\vdash\frac{1}{5}$ as result. That is, the result is $\frac{1}{5}$.

\subsection{Correctness of the algorithm}
\label{section:correctness}

\begin{theorem}
\label{theorem:algorithm-correct}
Let $t(x_1, \dots, x_n)$ be any {\L}ukasiewicz $\mu$-term.
Then, for every input vector $(r_1, \dots, r_n) \in [0,1]^n$, 
the above (Real RAM) algorithm terminates with a conditioned linear expression
$\GLE{C_{\vv{r}} \,}{\, e_{\vv{r}}}$ satisfying properties (P1) and (P2). 
Moreover, the set of all 
possible resulting conditioned linear expressions
\begin{equation}
\label{equation:result-set}
\{ \GLE{C_{\vv{r}} \,}{\, e_{\vv{r}}} \mid \vv{r} \in [0,1]^n \}
\end{equation}
is finite, and thus provides a representing system for the function $t \colon [0,1]^n \to [0,1]$.
\end{theorem}
\noindent
Before the proof it is convenient to introduce some terminology associated with the properties stated in the theorem.
For a $\mu$-term $t$, we call the cardinality of the set
(\ref{equation:result-set}) of possible results, $\GLE{C_{\vv{r}} \,}{\, e_{\vv{r}}}$,
the  \emph{basis size},
and we call the maximum number of inequalities in any $C_{\vv{r}}$ the 
\emph{condition size}.

\begin{proof} 
The proof is by induction on the structure of $t$. We verify the critical case in which $t$ is $\Mu{x_{n+1}}{t'}$.

We show first that the loop invariants (I1), (I2) guarantee that any result  returned via (\ref{equation:first-result}) or (\ref{equation:second-result}) satisfies (P1) and (P2).
By induction hypothesis, the recursive computation of 
$t'(x_1, \dots, x_{n+1})$ at $(\vv{r},d(\vv{r}))$ as $\GLE{C}{e}$, where 
$e$ has the form 
$q_1\, x_1 + \dots + q_n \, x_n + q_{n+1}\,  x_{n+1} + q$  as in (\ref{equation:e-form}), satisfies:
$C(\vv{r}, d(\vv{r}))$; and, for all
$s_1, \dots, s_{n+1} \in \mathbb{R}$, if 
$C(s_1, \dots, s_{n+1})$ then $\vv{s} \in [0,1]^n$ and 
$t'(s_1, \dots, s_{n+1}) = e(s_1, \dots, s_{n+1})$. 

In the case that $q_{n+1} \neq 1$, the linear expression $f$, defined in
(\ref{equation:f}), maps any $s_1, \dots, s_n \in \mathbb{R}$ to the unique solution
$f(\vv{s})$ to the equation $x_{n+1} = e(s_1, \dots, s_n, x_{n+1})$ in $\mathbb{R}$.
Suppose that $D(\vv{s})$ holds.
Then, by  loop invariant (I2), $d(\vv{s}) \leq (\Mu{x_{n+1}}{t'})(\vv{s})$.
Suppose also that $C(\vv{s}, f(\vv{s}))$. Then 
$t'(\vv{s},f(\vv{s})) = e(\vv{s},f(\vv{s})) = f(\vv{s})$, i.e., 
$f(\vv{s})$ is a fixed point of $x_{n+1} \mapsto t'(\vv{s}, x_{n+1})$; whence, 
$(\Mu{x_{n+1}}{t'})(\vv{s}) \leq f(\vv{s})$.
Suppose, finally, that $C(\vv{s}, d(\vv{s}))$ also holds. Then, because
both $C(\vv{s}, d(\vv{s}))$ and $C(\vv{s}, f(\vv{s}))$, and 
$d(\vv{s}) \leq (\Mu{x_{n+1}}{t'})(\vv{s}) \leq f(\vv{s})$, we have, by the convexity of constraints, that 
$t'(\vv{s}, s_{n+1}) = e(\vv{s}, s_{n+1})$ for all $s_{n+1} \in [d(\vv{s}),f(\vv{s})]$. So $f(\vv{s})$ is the unique fixed-point of
$x_{n+1} \mapsto t'(\vv{s}, x_{n+1})$ on $[d(\vv{s}),f(\vv{s})]$. Since, $d(\vv{s}) \leq (\Mu{x_{n+1}}{t'})(\vv{s})$, 
we have $f(\vv{s}) = (\Mu{x_{n+1}}{t'})(\vv{s})$. This argument justifies that the conditioned linear expression of 
(\ref{equation:first-result}) satisfies (P2). It satisfies
(P1) just if $C(\vv{r},f(\vv{r}))$, which is exactly the condition under which (\ref{equation:first-result}) is returned as the result.

In the case that $q_{n+1} = 1$ then, for any  $s_1, \dots, s_n \in \mathbb{R}$, 
the equation $x_{n+1} = e(s_1, \dots, s_n, x_{n+1})$  has a solution if and only if
$q_1 \, s_1 + \dots + q_n \, s_n  + q  = 0$, 
in which case any $x_{n+1} \in \mathbb{R}$ is a solution. 
Suppose that $q_1 \, s_1 + \dots + q_n \, s_n  + q  = 0$ and 
$C(\vv{s}, d(\vv{s}))$ both hold. Then $t'(s_1, \dots, s_n,d(\vv{s})) =
e(\vv{s}, d(\vv{s})) = d(\vv{s})$, so $d(\vv{x})$ is a fixed point of 
$x_{n+1} \mapsto t'(\vv{s}, x_{n+1})$. If also $D(\vv{s})$ holds then, by
 loop invariant (I2), $d(\vv{x}) = (\Mu{x_{n+1}}{t'})(\vv{s})$.
We have justified that the conditioned linear expression of 
(\ref{equation:second-result}) satisfies (P2). It satisfies
(P1) just if $q_1 \, r_1 + \dots + q_n \, r_n  + q  = 0$,
which is exactly the condition under which (\ref{equation:second-result}) is returned as the result.

Next we show that the loop invariants are preserved through the computation. 
Properties (I1) and (I2) are trivially satisfied by the initial values 
$D = \emptyset$ and $d =0$. We must show that they are preserved when $D$ and $d$ are modified via (\ref{equation:next-iteration}), which happens when 
execution passes to \textbf{find next approximation}.
In this subroutine, the inequalities  in $C$ are first arranged as in (\ref{equation:C}) where,
as $C(\vv{r},d(\vv{r}))$, we must have  $m \geq 1$, as otherwise $C(\vv{r},s)$ would hold for all real $s \geq  d(\vv{r})$, contradicting that $C(\vv{r},s)$ implies $s \in [0,1]$. (Similarly, $l \geq 1$.)
%
Thus there indeed exists $j$ with $1 \leq j \leq m$ such that 
$b_j(\vv{r}) \leq b_i(\vv{r})$ for all $i$ with $1 \leq i \leq m$. 
It is immediate that the constraints in the modified
$D$  of (\ref{equation:next-iteration}) 
are true for $\vv{r}$. Thus (I1) is preserved.
To show (I2), 
suppose $s_1, \dots, s_n$ satisfy the constraints, i.e.,
\begin{equation*}
D(\vv{s})  \qquad  C(\vv{s}, d(\vv{s})) \qquad
N(\vv{s}) \qquad 
\{b_j(\vv{s}) \leq b_i(\vv{s}) \mid 1 \leq i \leq m\} \enspace .
\end{equation*}
Defining $r' = (\Mu{x_{n+1}}{t'})(\vv{s})$, by (I2) for $D,d$ we have 
$d(\vv{s}) \leq r'$. We must show that $e(\vv{s},b_j(\vv{s})) \leq r'$.
By the definition of  $N(x_1, \dots, x_{n})$, in either the 
$q_{n+1} \neq  1$ or  $q_{n+1} = 1$ case, $N(\vv{s})$ implies that 
$C(\vv{s},\, r')$ does not hold.
Because $C(\vv{s},d(\vv{s}))$ and by the choice of $j$, it holds that 
$C(\vv{s},s)$, for all $s \in [0,1]$ such that  $s = d(\vv{s})$  or $d(\vv{s}) < s < b_j(\vv{s})$.
Since $C(\vv{s},r')$ is false and $d(\vv{s}) \leq r'$, it follows from the convexity of the conditioning set $C$ that, for every $s$
with  $s = d(\vv{s})$  or $d(\vv{s}) < s < b_j(\vv{s})$, we have $s < r'$.  
Whence, since $r'$ is the least prefixed point for $x_{n+1} \mapsto {t'}(\vv{s},x_{n+1})$, 
also $s < {t'}(\vv{s},s) \leq r'$, i.e., 
\begin{equation}
\label{equation:little}
s < e(\vv{s},s) \leq r' \enspace .
\end{equation}
Thus, $e(\vv{s},b_j(\vv{s})) = \sup \{ e(\vv{s},s) \mid 
\text{$s = d(\vv{s})$ or $d(\vv{s}) \leq s < b_j(\vv{s})$} \} \leq r'$.
Thus, $e(\vv{s},b_j(\vv{s})) \leq r'$, i.e., it is an approximation to the fixed point.
Moreover, it is a good new approximation to choose in the sense that:
\begin{equation}
\label{equation:better}
d(\vv{s}) < e(\vv{s},b_j(\vv{s}))~\text{~and ~ not}~
C(\vv{s}, e(\vv{s},b_j(\vv{s}))) \enspace .
\end{equation}
The former holds because $d(\vv{s}) < e(\vv{s},d(\vv{s}))$, by
(\ref{equation:little}),  and $d(\vv{s}) \leq b_j(\vv{s})$. The latter because if
$C(\vv{s}, e(\vv{s},b_j(\vv{s})))$ then, in particular, 
$e(\vv{s},b_j(\vv{s})) \leq b_j(\vv{s})$, so
$b_j(\vv{s}) = e(\vv{s},b_j(\vv{s})) = r'$, contradicting that 
not $C(\vv{s},r')$.

To show termination, by induction hypothesis, collecting all possible results of running the 
algorithm on $t'$ produces a representing system for $t' \colon [0,1]^{n+1} \to [0,1]$:
\begin{equation}
\label{equation:list}
\GLE{C_1\,}{\,e_1} \quad \dots \quad \GLE{C_{k'}\,}{\,e_{k'}} \enspace ,
\end{equation}
where $k'$ is the basis size of $t'$.
We now analyse the execution of the algorithm for $\Mu{x_{n+1}}{t'}$  on a given input vector $(r_1, \dots, r_n)$. On iteration number $i$, the loop is 
entered with constraints $D_i$ and approximation $d_i$ (where 
$D_1 = \emptyset$ and $d_1 = 0$), after which the recursive call to the 
algorithm for $t'$ yields one
of the conditioned linear expressions, $\GLE{C_{k_i}\,}{\,e_{k_i}}$, from
(\ref{equation:list}) above,
such that
$C_{k_i}(\vv{r}, d_i(\vv{r}))$ holds. Then, depending on conditions involving 
only $\GLE{C_{k_i}\,}{\,e_{k_i}}$ and $\vv{r}$, either a 
result is returned, or $D_{i+1}$ and $d_{i+1}$ are constructed for the 
loop to be repeated.
By (\ref{equation:better}), at iteration $i+1$ of the loop, we have
$d_{i+1} (\vv{r}) > d_i(\vv{r})$ and also $C_{k_i}(\vv{r}, d_{i+1}(\vv{r}))$ is false. 
Since each conditioning set
is convex, it follows that no $C_j$ can occur twice in the list $C_{k_1}, C_{k_2}, \dots$. Hence the algorithm must exit the loop after at most $k'$ iterations. Therefore, the computation for $\Mu{x}{t'}$ at $\vv{r}$ terminates. 

It remains to show that the algorithm for $\Mu{x}{t'}$  produces only finitely many conditioned linear expressions $\GLE{C_{\vv{r}}\,}{\,e_{\vv{r}}}$. 



We analyse the control flow in the algorithm for 
$\Mu{x_{n+1}}{t'}$  on a given input vector $(r_1, \dots, r_n)$. On iteration number $i$, the loop is 
entered with constraints $D_i$ and approximation $d_i$, after which 
the recursive call to the 
algorithm for $t'$ yields one
of the conditioned linear expressions, $\GLE{C_{k_i}\,}{\,e_{k_i}}$.
Suppose that $C_{k_i}$ and $D_i$ contain $u$ and $v$ inequalities respectively. If the loop is exited producing (\ref{equation:first-result}) as result then the resulting $C_{\vv{r}}$ has $2u +v$ inequalities. 
If it is exited producing (\ref{equation:second-result}) as result then 
$C_{\vv{r}}$ has $u +v+2$ inequalities (where $u + v + 2 \leq 2u+v$ because $C_{k_i}$ has to enforce the range constraint $0 \leq x_{n+1} \leq 1$). Otherwise, the algorithm repeats the loop, entering iteration $i+1$ with $D_{i+1}$, given by~(\ref{equation:next-iteration}),
having at most $2u + v$ inequalities ($N$ contributes $1$ inequality, and
there are at most $u-1$ inequalities  $b_j \leq b_i$ in (\ref{equation:next-iteration}) since $l \geq 1$). 

Therefore, if $l'$ is now the 
maximum number of inequalities occurring in any
$C_j$ from 
(\ref{equation:list}) (i.e., if it is the condition size for $t'$)
the algorithm for $\Mu{x_{n+1}}{t'}$ at $\vv{r}$, which runs for at most $k'$ iterations, results in $C_{\vv{r}}$ containing at most $2k'l'$ inequalities. 

To bound the number of  results $\GLE{C_{\vv{r}}}{e_{\vv{r}}}$, we count the possible control flows of the algorithm. At iteration $i$, the algorithm uses $\GLE{C_{k_i}\,}{\,e_{k_i}}$ from
(\ref{equation:list}), using which it might terminate with either 
(\ref{equation:first-result}) or (\ref{equation:second-result}), or it might
repeat the loop, 
entering iteration $i+1$ with $D_{i+1}$, given by~(\ref{equation:next-iteration}),
which can arise from $C_{k_i\,}$ in a number of ways determined by the possible
pairs of choices for $N$ and $b_j$ in (\ref{equation:next-iteration}).
In the case that the variable vector $(x_1, \dots, x_n)$ is empty (i.e., the term
$\Mu{x_{n+1}}{t'}$ is closed) the constraints in $D$ are redundant (they are simply true inequalities between rationals) and so can be discarded. 
In the case that $n \geq 1$, there are at least $2$ inequalities in $C$ giving range constraints on $x_1$, so there are at most $l'$ choices for $N$
($l'-2$ choices in the case that $q_{n+1} \neq 1$, and $2$ in the case $q_{n+1}= 1$).
Irrespective of $n$, there are at most $l'-1$ choices for $b_j$ (taking $n$ into account this can be improved to $l'-2n-1$).
Therefore, the execution
of the algorithm, is determined by the sequence:
\[
k_1, \,  u_1, \, k_2, \, u_2, \, \dots, \, k_m , \, v
\]
where: $m \leq k'$ is the number of loop iterations performed; each $u_i$, where  $1 \leq u_i \leq l'(l'-1)$, represents the choice of $N$ and $b_j$ used in the construction of $D_{i+1}$ (\ref{equation:next-iteration}), and $v$ is $ 1$ or $2$ according to whether 
the resulting $\GLE{C_{\vv{r}}\,}{\,e_{\vv{r}}}$ is  returned 
via (\ref{equation:first-result}) or (\ref{equation:second-result}). Since each
number $k_i$ is distinct, the number of different such sequences is bounded by:
\begin{equation}
\label{equation:bound}
2 \sum_{m = 1}^{k'} \frac{k'!}{(k'-m)!} (l'\,(l'-1))^{m-1} \; \leq \; (k'(l')^2)^{k'} \enspace ,
\end{equation}
where the right-hand-side gives a somewhat loose upper bound. Therefore, the number of possible results $\GLE{C_{\vv{r}}\,}{\,e_{\vv{r}}}$ for the algorithm for $\Mu{x_{n+1}}{t'}$ is at most $(k'(l')^2)^{k'}$.

\end{proof}

The above proof gives a truly abysmal complexity bound for the algorithm.
Let the basis and condition size for the term $t'(x_1, \dots , x_{n+1})$ be
$k'$ and $l'$ respectively. 
Then, as in the proof, the basis and condition size for
$\Mu{x_{n+1}}{t'}$ are respectively bounded by:
\[
k \: \leq\: (k'(l')^2)^{k'} ~~ \text{and}~~  l \: \leq \: 2k'l' \enspace .
\]
Using these bounds, the basis and condition size have non-elementary growth in the number of fixed points in a term $t$.

\subsection{Comparison}\label{section:comparison}

According to the crude complexity analyses we have given, the evaluation of 
{\L}ukasiewicz $\mu$-terms via rational linear arithmetic and (naive) black-box quantifier elimination is (in having doubly-\ and triply-exponential space and time complexity bounds) 
preferable to the (non-elementary space and hence time) evaluation via the direct algorithm. 
Nevertheless, we expect the direct algorithm to work better than this in practice. 
Indeed, a  main  motivating factor in the design of the direct algorithm 
is that the algorithm for $\Mu{x_{n+1}}{t'}$ 
only explores as much of the basis set for $t'$ as it needs to, and does so in an order that is tightly constrained by the monotone improvements made to the approximating $d$ expressions along the way. In contrast, the crude
complexity analysis is based on a worst-case scenario in which the
algorithm  is assumed to visit the entire basis for 
$t'$, and, moreover, to do so, for different input vectors $\vv{r}$, in every
possible order for  visiting the different basis sets. 
Perhaps better bounds can be obtained by a combination of relatively straightforward optimisations  of the algorithm and more careful analysis.

\section{ {\L}ukasiewicz  Modal $\mu$-Calculus}

In this section we introduce a modal extension of the {\L}ukasiewicz $\mu$-Calculus whose formulas are interpreted over  \emph{probabilistic nondeterministic transition systems} also known as \emph{Markov decision processes} (see, e.g., Section 10.6 of \cite{BaierKatoenBook}).

\begin{definition}
Given a set $S$ we denote with $\mathcal{D}(S)$ the set of \emph{(discrete) probability distributions} on $S$ defined as $\mathcal{D}(S)\!=\!\{Êd:S\rightarrow[0,1] \ | \ \sum_{s\in S} d(s) = 1\}$. We say that $d\!\in\!\mathcal{D}(S)$ is \emph{rational} if $d(s)$ is a rational number, for all $s\!\in\! S$.
\end{definition}
\begin{definition}
A \emph{probabilistic nondeterministic transition system} (PNTS) is a pair $(S, \rightarrow)$ where $S$ is a set of states and $\rightarrow\ \subseteq S\times \mathcal{D}(S)$ is the \emph{accessibility} relation. We write $s\not \rightarrow$ if $\{ d \ | \ s\rightarrow d\}\!=\!\emptyset$. A PNTS $(S,\rightarrow)$ is \emph{finite rational} if $S$ is finite and  $\bigcup_{s\in S}\{ d \ | \ s\rightarrow d\}$ is a finite set of rational probability distributions.
\end{definition}


We now introduce the  {\L}ukasiewicz  Modal $\mu$-Calculus which extends the probabilistic (or quantitative) modal $\mu$-calculus (qL$\mu$) of \cite{HM96,MM97,MM07,AM04}.

\begin{definition}
The syntax of formulas of {\L}ukasiewicz  Modal $\mu$-Calculus is generated by the following grammar:
\[
\phi \, ::= \,   X \mid P \mid \negate{P} \mid 
q \, \phi \mid  
 \phi \weakor \phi \mid \phi \weakand \phi \mid \phi \strongor \phi \mid \phi \strongand \phi  \mid \ \Diamond\phi \mid \ \Box \phi \mid  \Mu{X}{\phi} \mid \Nu{X}{\phi} \enspace ,\]
where $q$ ranges over rationals in $[0,1]$, $X$ over a countable set $\texttt{Var}$ of variables and $P$ over a set $\texttt{Prop}$ of propositional letters which come paired with associated complements $\negate{P}$. As a convention we denote with $\underline{1}$ the formula $\nu X.X$ and with $\underline{q}$ the formula $q\, \underline{1}$.
\end{definition}

Thus, beside the addition of the modal operators $\Diamond$ and $\Box$,  we also added atomic propositional letters and, to adhere with well-established notation, we used upper-case letters for the variables and greek letters for formulas.  Without introducing any significant ambiguity, we also write \Lukmu\ as a shorthand for {\L}ukasiewicz  Modal $\mu$-Calculus.  For mild convenience in the encoding of $\PCTL$ below, we consider a version with unlabelled modalities and propositional letters. However, the approach of this paper easily adapts to a labeled version of $\Lukmu$. 

The probabilistic (or quantitative) modal $\mu$-calculus (qL$\mu$) of \cite{HM96,MM97,MM07,AM04} is the fragment of \Lukmu\ obtained by removing the {\L}ukasiewicz connectives ($\odot$, $\oplus$) and \emph{scalar multiplications} ($q\, \phi$) by rationals numbers in $[0,1]$.

 {\L}ukasiewicz  Modal $\mu$-Calculus formulas are interpreted over PNTS's as we now describe.

\begin{definition}
\label{definition:interpretation}
Given a PNTS $(S,\rightarrow)$, an \emph{interpretation} for the variables and propositional letters is a function $\rho:(\texttt{Var}\uplus\texttt{Prop}) \rightarrow (S\rightarrow[0,1])$ such that $\rho(\negate{P})(x)\!=\! 1- \rho(P)(x)$. Given a function $f\!:\!S\rightarrow [0,1]$ and $X\!\in\!\texttt{Var}$ we define the interpretation  $\rho[f/X]$  as $\rho[f/X](X)\!=\!f$ and $\rho[f/X](Y)\!=\!\rho(Y)$, for $X\neq Y$.
\end{definition}

\begin{definition}
\label{definition:semantics}
The semantics of a $\Lukmu$ formula $\phi$ interpreted over $(S,\rightarrow)$ with interpretation $\rho$ is a function $\sem{\phi}_\rho: S\rightarrow[0,1]$ defined inductively on the structure of $\phi$ as follows:

\begin{center}
\begin{tabular}{ l l l}
$\sem{X}_{\rho}=\rho(X)$    & & $\sem{q\, \phi}_{\rho}(x)= q \cdot \sem{\phi}_\rho (x)$\\
$\sem{P}_{\rho}=\rho(P)$ & & $\sem{\negate{P}}_{\rho}= 1-\rho(P)$ \\
$\sem{\phi \sqcup \psi}_\rho(x)= \max \{ \sem{\phi}_{\rho}(x)Ê  ,  \sem{\psi}_{\rho}(x)  \}$ & & 
$\sem{\phi \sqcap \psi}_\rho(x)= \min \{ \sem{\phi}_{\rho}(x)Ê  ,  \sem{\psi}_{\rho}(x)  \}$ \\ 
$\sem{\phi \oplus \psi}_\rho(x)= \min \{ 1, \sem{\phi}_{\rho}(x)Ê  +  \sem{\psi}_{\rho}(x)  \}$ &&
$\sem{\phi \odot \psi}_\rho(x)= \max \{ 0, \sem{\phi}_{\rho}(x)Ê +   \sem{\psi}_{\rho}(x) -1  \}$ \\ 
$\sem{\Diamond \phi}_{\rho}(x)= \displaystyle \bigsqcup_{x\rightarrow d}\Big( \sum_{y\in S} d(y)\sem{\phi}_\rho(y)\Big)$&&
$\sem{\Box \phi}_{\rho}(x)= \displaystyle \bigsqcap_{x\rightarrow d}\Big(\sum_{y\in S} d(y)\sem{\phi}_\rho(y) \Big)$\\
$\sem{\mu X.\phi}_\rho = \lfp \big(  f \mapsto  \sem{\phi}_{\rho[f/X]}\big)$ & & 
$\sem{\mu X.\phi}_\rho = \gfp \big(  f \mapsto \sem{\phi}_{\rho[f/X]}\big)$\\
\end{tabular}
\end{center}
As in Section~\ref{section_lmu}, the interpretation of every operator is monotone, thus the existence of least and greatest points in the last two clauses is guaranteed by the the Knaster-Tarski theorem. 
\end{definition}

The operation $\Dual(\_ )$ (see Section \ref{section_lmu}) is extended to $\Lukmu$ modal formulas by defining $\Dual(\Diamond\phi)\!=\!\Box(\Dual(\phi))$ and $\Dual(\Box\phi)\!=\!\Diamond(\Dual(\phi))$ as expected.

\section{Encoding $\PCTL$ in {\L}ukasiewicz  Modal $\mu$-Calculus}

In this section, we show that the {\L}ukasiewicz  Modal $\mu$-Calculus of the previous section can encode the logic $\PCTL$. We refer the reader to Section 10.6.2 of \cite{BaierKatoenBook} for an extensive presentation of PCTL.


\subsection{Summary of $\PCTL$}

The notions of  \emph{paths}, \emph{schedulers} and \emph{Markov runs} in a PNTS are at the basis of the logic $\PCTL$. 
\begin{definition}\label{run_PLTS}
For a given PNTS $\lts=(S,\rightarrow)$ the binary relation $\leadsto_\lts \ \subseteq S\times S $ is defined as follows: $\leadsto_\lts= \{ (s,t) \ | \ \exists d. (s \rightarrow d \ \wedge \ d(t)>0)\}$. Note that $s\not \rightarrow$ if and only if $s\not\leadsto$.  We refer to $(S,\leadsto)$ as the \emph{graph underlying $\lts$}.
\end{definition}

\begin{definition}
A \emph{path} in a PNTS $\lts=(S,\rightarrow)$ is an ordinary path in the graph $(S,\leadsto)$, i.e., a finite or infinite sequence $\{s_i\}_{i\in I}$ of states such that $s_i\leadsto s_{i+1}$, for all $i+1\in I$. We say that a path is  \emph{maximal} if either it is infinite or it is finite and its last entry is a state $s_n$ without successors, i.e., such that $s_n\not \leadsto$. We denote with $\textnormal{P}(\lts)$ the set of all maximal paths in $\lts$. The set $\textnormal{P}(\lts)$  is endowed with the topology generated by the basic open sets $U_{\vv{s}}=\{  \vv{r}\ | \  \vv{s}\sqsubseteq \vv{r} \}$ where $\vv{s}$ is a finite sequence of states and $\sqsubseteq$ denotes the prefix relation on sequences. The space $\textnormal{P}(\lts)$ is always $0$-dimensional, i.e., the basic sets $U_{\vv{s}}$ are both open and closed and thus form a Boolean algebra. We denote with $\textnormal{P}(s)$ the open set $U_{\{s\}}$ of all maximal paths having $s$ as first state.
\end{definition}

\begin{definition}\label{scheduler_def}
A \emph{scheduler} in a PNTS $(S,\rightarrow)$ is a partial function $\sigma$ from non-empty finite sequences $s_0.\dots s_n$ of states to probability distributions $d\in\mathcal{D}(S)$ such that  $\sigma(s_0.\dots s_n)$ is not defined if and only if $s_n \not\rightarrow$ and, if $\sigma$ is defined at $s_0 . \dots s_n$ with $\sigma(s_0.\dots s_n)=d$, then $s_n\rightarrow d$ holds. A pair $(s,\sigma)$ is called a \emph{Markov run} in $\mathcal{L}$ and denoted by $M^s_\sigma$. It is clear that each Markov run $M^s_\sigma$ can be identified with a (generally) infinite Markov chain (having a tree structure) whose vertices are finite sequences of states and having $\{s\}$ as root.
\end{definition}

Markov runs are useful as they naturally induce probability measures on the space $\textnormal{P}(\mathcal{L})$.

\begin{definition}
Let $\mathcal{L}=(S,\rightarrow)$ be a PNTS and $M^s_\sigma$ a Markov run. We define the measure $m^s_\sigma$ on $\textnormal{P}(\mathcal{L})$ as the unique (by Carath{\'e}odory extension theorem) measure specified by the following assignment of basic open sets:
\begin{center}
$\displaystyle m^s_\sigma\big( \textnormal{U}_{s_0.\dots s_n} \big) = \prod^{n-1}_{i=0} d_i (s_{i+1})$
\end{center} 
 where $d_i=\sigma(s_0.\dots s_i)$ and $\prod\emptyset = 1$. It is simple to verify that $m^s_\sigma$ is a probability measure, i.e., $m^s_\sigma(\textnormal{P}(\mathcal{L}))=1$. We refer to $m^s_\sigma$ as the probability measure on $\textnormal{P}(\mathcal{L})$ induced by the Markov run $M^s_\sigma$.
\end{definition}

We are now ready to specify the syntax and semantics of  $\PCTL$.
\begin{definition}
Let the letter $P$ range over a countable set of propositional symbols $\texttt{Prop}$. The class of  $\PCTL$ \emph{state-formulas} $\phi$ is generated by the following two-sorted grammar:
\begin{center}
$ \phi::= \  \textnormal{true} \ | \ P \ | \ \neg \phi \ | \ \phi \vee \phi \ |  \ \exists  \psi\ |  \ \forall  \psi\ | \ \mathbb{P}^\exists_{\rtimes q}\psi \ | \ \ \mathbb{P}^\forall_{\rtimes q}\psi$
\end{center}
with $q\in\mathbb{Q}\cap[0,1]$ and  $\rtimes\in\{ >, \geq\}$, where \emph{path-formulas} $\psi$ are generated by the simple grammar: $\psi::= \  \circ \phi \ | \ \phi_1 \mathcal{U}\phi_2$. Adopting standard terminology, we refer to the connectives $\circ$ and $\mathcal{U}$ as the \emph{next} and  \emph{until}  operators, respectively.
\end{definition}

\begin{definition}
Given a PNTS $(S,\rightarrow)$, a \emph{$\PCTL$-interpretation} for the propositional letters is a function $\rho:\texttt{Prop}\rightarrow 2^S$, where $2^S$ denotes the powerset of $S$. 
\end{definition}

\begin{definition}
Given a PNTS $(S,\rightarrow)$ and a  $\PCTL$-interpretation $\rho$ for the propositional letters, the semantics $\gsem{\phi}_\rho$ of a  $\PCTL$ state-formula $\phi$ is a subset of $S$ 
\begin{itemize}
\item $\gsem{ \textnormal{true} }_\rho= S$,
$\gsem{P}_\rho = \rho(P)$,
$\gsem{\phi_1 \vee \phi_2}_{\rho}= \gsem{\phi_1}_\rho \cup \gsem{\phi_2}_{\rho}$,
$\gsem{\neg \phi }_{\rho}= S\setminus \gsem{\phi}_{\rho}$,
 \item $s\in \gsem{\exists \psi }_{\rho}$ if and only there exists $\vv{s}\in \textnormal{P}(s)$ such that that $\vv{s}\in \sem{\psi}_\rho$
  \item $s\in\gsem{\forall \psi }_{\rho}$ if and only for all $\vv{s}\in \textnormal{P}(s)$ it holds that $\vv{s}\in \gsem{\psi}_\rho$
 \item $s\in\gsem{\mathbb{P}^\exists_{\rtimes q} \psi }_{\rho}$ if and only $\big(\bigsqcup_{\sigma} m^s_\sigma (\gsem{\psi}_\rho)\big) \rtimes q$
  \item $s\in\gsem{\mathbb{P}^\forall_{\rtimes q} \psi }_{\rho}$ if and only $\big(\bigsqcap_{\sigma} m^s_\sigma (\gsem{\psi}_\rho)\big) \rtimes q$
\end{itemize}
where $\sigma$ ranges over schedulers and the semantics $\gsem{\psi}_\rho$ of path formulas is a subset of $\textnormal{P}(\mathcal{L})$ 
defined as:
\begin{itemize}
\item $\vv{s} \in \gsem{\circ \phi}_\rho $ if and only if $| \vv{s} | \geq 2$ (i.e., $\vv{s}=s_0.s_1.\dots$) and $s_1\in \gsem{\phi}_\rho$,
\item $\vv{s} \in \gsem{\phi_1 \mathcal{U} \phi_2}_\rho$ if and only if  $\exists n.\big( (s_n\in\gsem{\phi_2}_\rho) \wedge \forall m< n. (s_m\in\gsem{\phi_1}_\rho)\big)$,
\end{itemize}
\end{definition}
It is simple to verify that, for all path-formulas $\psi$, the set $\gsem{\psi}_\rho$ is Borel measurable \cite{BaierKatoenBook}. Therefore the definition is well specified. Note how the logic  $\PCTL$ can express probabilistic properties, by means of the connectives $\mathbb{P}^\forall_{\rtimes q}$ and $\mathbb{P}^\exists_{\rtimes q}$, as well as (qualitative) properties of the graph underlying the PNTS by means of the quantifiers $\forall$ and $\exists$.


\subsection{The encoding of  $\PCTL$}\label{section_encoding}

As it turns out, the crucial property of $\Lukmu$ which allow the encoding of $\PCTL$ is the possibility of defining the derived threshold modalities of Definition \ref{thresholds_encoding} which, crucially, are not expressible in the probabilistic $\mu$-calculus (qL$\mu$) of  \cite{HM96,MM97,MM07,AM04}. We reformulate here the result of Proposition \ref{semantics_threshold1} which straightforwardly extends to the present setting.

\begin{proposition}\label{semantics_threshold}
Let $(S,\rightarrow)$ be a PNTS, $\phi$ a $\Lukmu$ formula and $\rho$ an interpretation of the variables. Then it holds that:
\begin{center}
$\sem{\mathbb{P}_{\rtimes q}\phi}_\rho(s)= \left\{     \begin{array}{l  l}
 						1 & $if $\sem{\phi}_\rho (s) \rtimes q \\
						0 & $otherwise$\\
						\end{array}      \right.$
\end{center}
\end{proposition}

The following lemma is also useful.

\begin{lemma}\label{lemma_leadsto}
Let $(S,\rightarrow)$ be a PNTS, $\phi$ a $\Lukmu$ formula and $\rho$ an interpretation of the variables. Then:
\begin{itemize}
\item $\sem{\mathbb{P}_{>0}(\Diamond X)}_\rho (s) = 1$ iff\ $\ \exists t. \big(s\leadsto t \wedge \rho(X)(t)>0\big)$
\item $\sem{\mathbb{P}_{=1}(\Box X)}_\rho (s)=1$ iff\ $\  \forall t. \big(s\leadsto t \rightarrow \rho(X)(t)=1\big)$
\end{itemize}
\end{lemma}
\begin{proof}
Note that $\sem{\Diamond X}_\rho(s)>0$ holds if and only if there exists $s\rightarrow d$ such that $\sum_{t\in S} d(t)\rho(X)(t)>0$ holds. This is the case if and only if $d(t)\!>\!0$ (i.e., $s\leadsto t$) and $\rho(X)(t)\!>\!0$, for some $t\!\in\! S$. The result then follows by Proposition \ref{semantics_threshold}. The case for $\mathbb{P}_{=1}(\Box X)$ is similar.
\end{proof}

\begin{remark}\label{remark_leadsto}
When considering $\{0,1\}$-valued interpretations for $X$, the macro formula $\mathbb{P}_{>0}\Diamond$ expresses the meaning of the diamond modality in classical modal logic with respect to the graph $(S,\leadsto)$ underlying the PNTS. Similarly, $\mathbb{P}_{=1}\Box$ corresponds to the the classical box modality.
\end{remark}

We are now ready to define the encoding of  $\PCTL$ into $\Lukmu$.

\begin{definition}\label{ENCODING_PCTL_E}
We define the encoding $\mathbf{E}$ from  $\PCTL$ formulas to closed $\Lukmu$ formulas (where $\boxdot\phi$ stands for the $\Lukmu$ formula $\Box\phi \sqcap \Diamond\underline{1}$), by induction on the structure of the  $\PCTL$ formulas $\phi$ as follows:
\begin{enumerate}
\item $\label{enc_case_prop}\mathbf{E}(P)=P$,
\item $\mathbf{E}(\textnormal{true})\!=\! \underline{1}$,
\item $\mathbf{E}(\phi_1\vee \phi_2)= \mathbf{E}(\phi_1) \sqcup \mathbf{E}(\phi_2)$,
\item \label{enc_case_neg} $\mathbf{E}(\neg \phi)= \Dual( \mathbf{E}(\phi))$,
\item \label{enc_case_exists_circ}$\mathbf{E}(\exists( \circ  \phi))= \mathbb{P}_{>0}\big(\Diamond\mathbf{E}(\phi)\big)$, 
\item \label{enc_case_forall_circ}$\mathbf{E}(\forall( \circ  \phi))= \mathbb{P}_{=1}\big(\boxdot\mathbf{E}(\phi)\big)$, 
\item \label{enc_case_exists} $\mathbf{E}(\exists(\phi_1 \ \mathcal{U} \ \phi_2))= \mu X. \Big(\mathbf{E}(\phi_2) \sqcup \big(  \mathbf{E}(\phi_1) \sqcap \mathbb{P}_{>0}(\Diamond X)\big)\Big)$,
\item \label{enc_case_forall} $\mathbf{E}(\forall(\phi_1 \ \mathcal{U} \ \phi_2) )=  \mu X. \Big(\mathbf{E}(\phi_2) \sqcup \big(  \mathbf{E}(\phi_1) \sqcap \mathbb{P}_{=1}(\boxdot X)\big)\Big)$,
\item \label{case_threshold_1} $\mathbf{E}(\mathbb{P}^{\exists}_{\rtimes q }(\circ \phi))=\mathbb{P}_{\rtimes q} \big( \Diamond\mathbf{E}(\phi)\big)$,
\item \label{case_threshold_2}$\mathbf{E}(\mathbb{P}^{\forall}_{\rtimes q }(\circ \phi))=\mathbb{P}_{\rtimes q} \big( \boxdot\mathbf{E}(\phi)\big)$,
\item \label{case_threshold_3}$\mathbf{E}(\mathbb{P}^\exists_{\rtimes q}(\phi_1 \mathcal{U} \phi_2))= \mathbb{P}_{\rtimes q} \Big( \mu X. \Big(\mathbf{E}(\phi_2) \sqcup \big(  \mathbf{E}(\phi_1) \sqcap \Diamond X\big)\Big) \Big)$,
\item \label{case_threshold_4}$\mathbf{E}(\mathbb{P}^{\forall}_{\rtimes q}(\phi_1 \mathcal{U} \phi_2))= \mathbb{P}_{\rtimes q} \Big( \mu X. \Big(\mathbf{E}(\phi_2) \sqcup \big(  \mathbf{E}(\phi_1) \sqcap \boxdot X\big)\Big) \Big)$,
\end{enumerate}
Note that Case \ref{enc_case_neg} is well defined since $ \mathbf{E}(\phi)$ is closed by construction.
\end{definition}

\begin{remark}\label{remark_fragment_plmu}
The only occurrences of \L ukasiewicz operators $\{\oplus,\odot\}$ and scalar multiplication $(q\, \phi)$ in encoded $\PCTL$ formulas appear in the formation of the threshold modalities $\mathbb{P}_{\rtimes q}(\_ )$. Thus, $\PCTL$ can be also seen as a fragment of qL$\mu$ extended with threshold modalities as primitive operations. With the aid of these modalities the encoding is, manifestly, a straightforward adaption of the standard encoding of CTL into the  modal $\mu$-calculus (see, e.g., \cite{Stirling96}). 
\end{remark}

The main result of this section establishes the correctness of our encoding.
For convenience in the formulation, we identify subsets  with their corresponding characteristic 
functions. 

\begin{theorem}
\label{theorem:pctl}
For every PNTS $(S,\rightarrow)$,  $\PCTL$-interpretation $\rho\!:\!\textnormal{Prop}\rightarrow(S\rightarrow\!\{0,1\})$ of the propositional letters and  $\PCTL$ formula $\phi$, the equality $\gsem{\phi}_\rho(s)=\sem{\mathbf{E}(\phi)}_\rho(s)$ holds, for all $s\in S$.
\end{theorem}
\begin{proof}
The proof goes by induction on the structure of $\phi$. Cases \ref{enc_case_prop}--\ref{enc_case_neg} of Definition \ref{ENCODING_PCTL_E} are trivial. Case \ref{enc_case_exists_circ} follows directly from Lemma \ref{lemma_leadsto}. Observing that $\sem{\boxdot\phi}_\rho(s)=0$ if $s\not\leadsto$ and $\sem{\boxdot\phi}_\rho(s)=\sem{\Box \phi}_\rho(s)$ otherwise, the correctness of Case \ref{enc_case_forall_circ} directly follows from Lemma \ref{lemma_leadsto}. Consider cases \ref{enc_case_exists} and \ref{enc_case_forall}. The encoding is of the form $\mu X. (F \sqcup (G \sqcap H(X))$, where $F$ and $G$ (by induction hypothesis) and $H(X)$ (by Proposition \ref{semantics_threshold}) are all $\{0,1\}$-valued. Therefore the functor $f\mapsto \sem{F\sqcup (G\sqcap H(X))}_{\rho[f/X]}$ maps $\{0,1\}$-valued functions to $\{0,1\}$-valued functions and has only $\{0,1\}$-valued fixed-points. It then follows from Remark \ref{remark_leadsto} that the correctness of  these two cases can  be proved with the standard technique used to prove the correctness of the encoding of CTL into Kozen's $\mu$-calculus (see, e.g., \cite{Stirling96}).
Consider Case \ref{case_threshold_1}. It is immediate to verify that $\bigsqcup_{\sigma}\{ m^s_\sigma( U)\}$, where $U=\gsem{\circ \phi}_\rho=\bigcup\{ U_{\{s.t\}} \ | \ t\in \gsem{\phi}_\rho\}$, is equal (by induction hypothesis) to $\sem{\Diamond \mathbf{E}(\phi)}_\rho(s)$. The desired equality $\gsem{\mathbb{P}^{\exists}_{\rtimes q}\circ\phi}_{\rho}=\sem{\mathbb{P}_{\rtimes q }\Diamond\mathbf{E}(\phi)}_\rho$ then follows by Proposition \ref{semantics_threshold}. Case \ref{case_threshold_2} is similar. The two cases \ref{case_threshold_3} and \ref{case_threshold_4} are similar, thus we just consider case \ref{case_threshold_3}. Let $\phi= \mathbb{P}^\exists_{\rtimes q}(\psi)$ and $\psi= \phi_1 \mathcal{U} \phi_2$.  We denote with $\Psi$ the set of paths $\gsem{\psi}_\rho$. Denote by $F(X)$ the formula $\mathbf{E}(\phi_2) \sqcup (  \mathbf{E}(\phi_1) \sqcap \Diamond X)$.  By application of Lemma \ref{semantics_threshold} it is  sufficient to prove that the equality 
\begin{equation}\label{eq_1_proof_correctness}\bigsqcup_\sigma \{ m^s_\sigma(\Psi)\} = \sem{\mu X. F( X)\big) }_\rho(s)
\end{equation} holds. Note that $\mu X.F(X)$ can be expressed as an equivalent qL$\mu$ formula by substituting the closed subformulas $\mathbf{E}(\phi_1)$ and $\mathbf{E}(\phi_2)$ with two fresh atomic predicates $P_{i}$ with interpretations $\rho(P_i)=\sem{\mathbf{E}(\phi_i)}$. The equality can then be proved by simple arguments based on the game-semantics of qL$\mu$ (see, e.g., \cite{MM07} and \cite{MIO2012a}), similar to the ones used to prove that  the Kozen's $\mu$-calculus formula $\mu X. (P_2 \vee ( P_1 \wedge \Diamond X))$ has the same denotation of the CTL formula $\exists( P_1 \mathcal{U} P_2)$ (see, e.g., \cite{Stirling96}). However, to keep the paper self-contained,  we provide instead a direct proof of Equation \ref{eq_1_proof_correctness} based on denotational semantics.

 We first show that the inequality
\begin{equation}\label{proof_eq_1}
 \bigsqcup_\sigma \{ m^s_\sigma(\Psi)\} \leq \sem{\mu X.F(X)\big) }_\rho(s)
\end{equation}
holds. Define $\Psi_k\!=\!\{ {s_0.s_1.s_2\dots} \ | \ s_0=s \textnormal{ and }\exists n\leq k . \big( s_n\in \gsem{\phi_2}_\rho \wedge \forall m<n.( s_m\in\gsem{\phi_1}_\rho ) \big)\}$. Clearly $\Psi=\bigcup_k \Psi_k$. Suppose Inequality \ref{proof_eq_1} does not hold. Then there exists some $k$ and scheduler $\sigma$ such that 
\begin{equation}\label{proof_eq_2}
m^s_\sigma(\Psi_k)> \sem{\mu X. F(X) }_\rho(s)
\end{equation}
We prove that this is not possible by induction on $k$.  In the $k=0$ case, since we are assuming $m^s_\sigma(\Psi_0)>0$, it holds that $s\!\in\!\gsem{\phi_2}_\rho$. By inductive hypothesis on $\phi_2$, we know that $\sem{\mathbf{E}(\phi_2)}_\rho(s)=1$ and this implies that $\sem{\mu X.F(X)}_{\rho}(s)\!=\!1$, which is a contradiction with Inequality \ref{proof_eq_2}.  Consider the case $k+1$. Note that if $s\!\in\!\gsem{\phi_2}_\rho$ then $\sem{\mu X.F(X) }_\rho(s)\!=\!1$ as before, contradicting Inequality \ref{proof_eq_2}. So assume $s\!\not\in\!\gsem{\phi_2}_\rho$. Since we are assuming $m^s_\sigma(\Psi_{k+1})\!>\!0$ it  must be the case that $s\!\in\!\gsem{\phi_1}_\rho$. Similarly, $m^s_\sigma(\Psi_{k+1})>0$ and $s\not\in\gsem{\phi_2}_\rho$ imply that $s\not\rightarrow$ does not hold. This means (see Definition \ref{scheduler_def}) that $\sigma(\{s\})$ is defined. Let $d=\sigma(\{s\})$ and observe that $m^s_\sigma(\Psi_{k+1})\!=\! \sum_{t\in S}d(t)m^t_{\sigma^\prime}(\Psi_k)$, where $\sigma^\prime(s_0,s_1,\dots,s_n)=\sigma(s,s_0,s_1,\dots,s_n)$.  By induction on $k$ we know that the inequality $m^t_{\sigma^\prime}(\Psi_k)\leq \sem{\mu X.F(X) }_\rho(t)$ holds for every $t\! \in\! S$. Thus, by definition of the semantics of $\Diamond$, we obtain
$ m^s_\sigma( \Psi_k) \leq \sem{\Diamond \big( \mu X.F(X) \big) }_\rho(s)$. 
Recall that we previously assumed $s\not\in\gsem{\phi_2}_\rho$ and $s\in\gsem{\phi_1}_\rho$. Hence the equality
$$\sem{\Diamond \big( \mu X.F(X) \big) }_\rho(s) = \sem{\mathbf{E}(\phi_2) \sqcup ( \mathbf{E}(\phi_1) \sqcap \big( \Diamond \mu X.F(X) \big) )  }_{\rho}(s)$$
holds. The formula on the right is just the unfolding $F(\mu X.F(X))$ of $\mu X.F(X)$. This implies the desired contradiction with Inequality \ref{proof_eq_2}.

We now prove that also the opposite inequality
\begin{equation}\label{proof_eq_3}
 \bigsqcup_\sigma \{ m^s_\sigma(\psi)\} \geq \sem{\mu X. F(X) }_\rho(s)
\end{equation}
holds.  By Knaster-Tarski theorem, $\sem{\mu X.F(X)}_\rho=\bigsqcup_\alpha \sem{F(X)}_{\rho^{\alpha}}$, where $\alpha$ ranges over the ordinals and $\rho^{\alpha}\!=\!\rho[\bigsqcup_{\beta<\alpha} \sem{F(X)}_{\rho^\beta}/ X]$. We prove Inequality \ref{proof_eq_3} by showing, by transfinite induction, that for every ordinal $\alpha$ and $\epsilon >0$,  the inequality 
\begin{equation}\label{proof_eq_4}
\bigsqcup_\sigma \{ m^s_\sigma(\psi)\} > \sem{F(X) }_{\rho^\alpha}(s)-\epsilon
\end{equation} 
holds, for all $s\in S$.  The case for $\alpha=0$ is immediate since $\sem{F(X)}_{\rho^0}(s)\!>\!0$ if and only if $\sem{\mathbf{E}(\phi_2)}_\rho(s)\!=\!1$ and this implies  $\bigsqcup_\sigma \{ m^s_\sigma(\psi)\} \!=\!1$. Consider  the case $\alpha\!=\!\beta+1$. If $\sem{\mathbf{E}(\phi_2)}_\rho(s)\!=\!1$ then Inequality \ref{proof_eq_3} holds as above. Thus assume $\sem{\mathbf{E}(\phi_2)}_\rho(s)\!=\!0$. Note that $\sem{F(X)}_{\rho^0}(s)>0$ only if $\sem{\mathbf{E}(\phi_1)}(s)\!=\!1$. Thus assume $\sem{\mathbf{E}(\phi_1)}_\rho^\beta (s)\!=\!1$. Under these assumption,  $\sem{F(X)}_{\rho^{\alpha}}= \sem{\Diamond F(X)}_{\rho^\beta}$ as it is immediate to verify. By definition of the semantics of $\Diamond$ we have: 
$$
\sem{\Diamond F(X)}_{\rho^\beta}(s)= \bigsqcup_{s\rightarrow d} \big( \displaystyle \sum_{t\in S} d(t)\sem{F(X)}_{\rho^\beta}(t)\big)
$$
By induction hypothesis on $\beta$ we know that for every $t\!\in\! S $ and $\epsilon \!>\!0$ the inequality $\sem{F(X)}_\rho^{\beta}(t)\!<\! \bigsqcup_\sigma \{ m^t_\sigma(\psi)\} + \epsilon$ holds. Therefore the inequality  
$$
\sem{\Diamond F(X)}_{\rho^\beta}(s)  < \bigsqcup_{s\rightarrow d} \big( \displaystyle \sum_{t\in S} d(t)\Big( \bigsqcup_\sigma \{ m^t_\sigma(\psi)\} + \epsilon \Big) \big)
$$
holds for every $\epsilon \!>\!0$. For each transition $s\!\rightarrow\! d$ and scheduler $\sigma$ define $\sigma^d$ as $\sigma^d(\{s\})\!=\!d$ and $\sigma^d(s.t_0.\dots)=\sigma(t_0\dots)$. It follows from unfolding of definitions that that 
$$
\bigsqcup_{s\rightarrow d} \big( \displaystyle \sum_{t\in S} d(t)\Big( \bigsqcup_\sigma \{ m^t_\sigma(\psi)\} + \epsilon \Big) \big) =  \bigsqcup_{\sigma} \{ m^s_{\sigma^{d}}(\psi)\} + \epsilon
$$
and this conclude the proof of Inequality \ref{proof_eq_4} for the case $\alpha\!=\!\beta+1$. Lastly, the case for $\alpha$ a limit ordinal follows straightforwardly from the inductive hypothesis on $\beta<\alpha$.

\end{proof}


\section{Model checking}
\label{section:model-checking}

The \emph{model checking} (or formula evaluation) \emph{problem} of the \L ukasiewicz modal $\mu$-calculus is the following: given a closed $\Lukmu$ formula  $\phi$, a finite rational transition system $(S,\rightarrow)$ and a state $s\!\in\!S$, compute the value $\sem{\phi}(s)$.

In this section we reduce this problem to the evaluation of  \L ukasiewicz $\mu$-terms. We do this by effectively producing a closed rational $\mu$-term $t_s(\phi)$, with the property that $t_s (\phi)\! =\! \sem{\phi}(s)$, whence the rational value of $\sem{\phi}(s)$ can be calculated; for example, by the algorithm in Section~\ref{section:algorithm}. The construction of $t_s(\phi)$ is similar in spirit to 
the reduction of modal $\mu$-calculus model checking to a system of nested boolean fixed-point equations in Section 4 of~\cite{Mader1995}.

We assume, without loss of generality, that all fixed-point operators in $\phi$ bind distinct variables. Let $X_1, \dots, X_m$ be the variables appearing in $\phi$.
We write $\sigma_i \, X_i. \, \psi_i$ for the unique subformula of $\phi$ in which $X_i$ is bound. The strict (i.e., irreflexive) \emph{domination} relation $X_i \dominates X_j$ between variables is defined to mean that $\sigma_j \, X_j. \, \psi_j$ occurs as a subformula in $\psi_i$.

Suppose $|S| = n$. For each $s \in S$, we translate $\phi$ to a $\mu$-term $t_s(\phi)$ containing at most $mn$ variables $x_{i,s'}$, where $1 \leq i \leq m$ and $s' \in S$.
The translation is defined using a more general function
$t^\Gamma_s$, defined on subformulas of $\phi$, where 
$\Gamma \subseteq
\{1, \dots, m\} \times S$ is an  auxiliary component keeping track of the states at which variables have previously been encountered. Given $\Gamma$ and $(i,s) \in \{1, \dots, m\} \times S$, we define:
\[
\Gamma \dominates (i,s) \; = \; 
(\Gamma \cup \{(i,s)\}) \backslash
\{ (j,s') \in \Gamma \mid X_i \dominates X_j\} \enspace .
\]
This operation is used in the definition of Figure \ref{reduction_model_checking_figure} 
to `reset' subordinate fixed-point variables
whenever a new variable that dominates them is declared. 
\begin{figure}
\begin{align*}
t^\Gamma_s (X_i) & =  
       \begin{cases} 
           x_{i,s} & \text{if $(i,s) \in \Gamma$} \\
           \sigma_i \,  x_{i,s} . \; t^{\Gamma \dominates (i,s)}_s(\psi_i) & \text{if $(i,s) \notin \Gamma$} 
       \end{cases}
\\
t^\Gamma_s (P) & =  \underline{\rho(P)(s)}
\\
t^\Gamma_s (\negate{P}) & =  \underline{1-\rho(P)(s)}
\\
t^\Gamma_s (q \, \phi) & =  q \, t^\Gamma_s (\phi)
\\
t^\Gamma_s (\phi_1 \bullet \phi_2) & = 
     t^\Gamma_s (\phi_1) \bullet t^\Gamma_s (\phi_2) \qquad \bullet \in \{\weakor, \weakand, \strongor, \strongand\}
\\
t^\Gamma_s (\Diamond \phi) & = 
  \bigsqcup_{s \rightarrow d} ~ \bigoplus_{s' \in S} ~ d(s') \; t^\Gamma_{s'} (\phi) 
\\
t^\Gamma_s (\Box \phi) & = 
  \bigsqcap_{s \rightarrow d} ~ \bigoplus_{s' \in S} ~ d(s') \; t^\Gamma_{s'} (\phi) 
\\
t^\Gamma_s (\sigma_i \, X_i. \, \psi_i) & = 
  \sigma_i \,  x_{i,s} . \; t^{\Gamma \cup \{(i,s)\}}_s(\psi_i) 
\end{align*}
\caption{Reduction of the Model-Checking problem to \L ukasiewicz $\mu$-terms evaluation.}
\label{reduction_model_checking_figure}
\end{figure}
\noindent
This is well defined because changing from
$\Gamma$ to $\Gamma \dominates (i,s)$ or to $\Gamma \cup \{(i,s)\}$ strictly increases the function 
\begin{equation}
\label{equation:function}
i \mapsto | \{ (i,s) \mid (i,s) \in \Gamma\}| \colon \{1, \dots, m\} \to \{0, \dots, n\}
\end{equation}
under the lexicographic order on functions relative to $\dominates$.

For a 
$\mu$-term $t'$ containing variables $x_{i,s'}$ as above, we write $\FVI(t')$ 
for the set of pairs in the product set $\{1, \dots, m\} \times S$ that index free variables in 
$t'$. The translation defined in Figure~\ref{reduction_model_checking_figure} clearly satisfies: 
\begin{lemma}
Suppose  $\phi'$ is a subformula of $\phi$ and $\Gamma \subseteq \{1, \dots, m\} \times S$.
Then  $\FVI(t^\Gamma_s(\phi')) \subseteq \Gamma$.
\end{lemma}
Also, the translation can be defined for any finite PNTS, whether rational or not. When the PNTS is rational, so is the resulting $\mu$-term.

Now consider any $\Gamma \subseteq \{1, \dots, m\} \times S$. By a \emph{$\Gamma$-environment}, we mean any function $\gamma \colon \Gamma \to [0,1]$. We use the function $\gamma$ to generate (the variable part of) an 
interpretation $\rho^\gamma \colon \{1, \dots, m\} \rightarrow (S\rightarrow[0,1])$, as in Definition~\ref{definition:interpretation}, but with variables \texttt{Var} replaced with their indices. This is defined to be the unique interpretation that satisfies the equation below.
\[
\rho^\gamma (i) \; = \; 
\begin{cases}
\lfp \left({f \mapsto  s \mapsto 
\begin{cases} \gamma(i,s) & \text{if $(i,s) \in \Gamma$} \\
              \sem{\psi_i}_{\rho^\gamma[f/X_i]} & \text{if $(i,s) \notin \Gamma$} 
\end{cases}}
\right) & \text{if $\sigma_i = \mu$}
\\[4ex]
\gfp \left({f \mapsto  s \mapsto 
\begin{cases} \gamma(i,s) & \text{if $(i,s) \in \Gamma$} \\
              \sem{\psi_i}_{\rho^\gamma[f/X_i]} & \text{if $(i,s) \notin \Gamma$} 
\end{cases}}
\right) & \text{if $\sigma_i = \nu$}
\end{cases}
\]
Such a unique interpretation exists because the value of $\sem{\psi_i}_{\rho}$ depends only on the interpretation of $\rho$ on $X_i$ and variables that strictly dominate it.

\begin{lemma}
Let $\phi'$ be a subformula of $\phi$. Suppose $\Gamma \subseteq \{1, \dots, m\} \times S$ satisfies the property that, for every $(i,s) \in \Gamma$, the variable $X_i$  is not bound in $\phi'$. Let $\gamma$ be any $\Gamma$-environment. Then, for all $s \in S$, it holds that
$t^\Gamma_s(\phi')(\gamma) =  \sem{\phi'}_{\rho^\gamma}(s)$, where on the right hand side of the equality we are using $\gamma$ to assign values to the 
$x_{i,s}$ variables in $t^\Gamma_s(\phi')$ in the obvious way.
\end{lemma}
\begin{proof}
The lemma is proved by an outer induction on the lexicographic order on the function (\ref{equation:function}) determined by $\Gamma$ (as induction hypothesis, we assume the property holds for higher values of this function),
and by an inner induction on the structure of $\phi'$.
We consider the two interesting cases: when $\phi'$ is a variable, and when it is a fixed-point formula.

We treat first the case in which $\phi'$ is $X_i$. If $(i,s) \in \Gamma$ then 
we simply have:
\[
t^\Gamma_s(X_i)(\gamma) \; = \; x_{i,s}(\gamma) \; = \; \gamma(i,s) \; = \; \sem{X_i}_{\rho^\gamma}(s) 
\enspace .
\]
If $(i,s) \notin \Gamma$ then
$t^\Gamma_s(X_i)\, = \,
\sigma_i \,  x_{i,s} . \; t^{\Gamma \dominates (i,s)}_s(\psi_i)$. 
Without loss of generality, we consider the case that $\sigma$ is a least-fixed-point $\mu$. 
For $v \in [0,1]$, we define $\gamma \dominates [v/(i,s)]$ to be the $(\Gamma\dominates (i,s))$-environment that assigns $v$ to $(i,s)$ and agrees with $\gamma$ on $\Gamma \cap (\Gamma\dominates (i,s))$, and we write $\gamma'$ for the restriction of $\gamma$ (equivalently of $\gamma \dominates [v/(i,s)]$) to $\Gamma \cap (\Gamma\dominates (i,s))$. Then we have:
\begin{align*}
t^\Gamma_s(X_i)(\gamma)\; & = \; (\mu   x_{i,s} . \; t^{\Gamma \dominates (i,s)}_s(\psi_i))(\gamma) \\
& = \; \lfp (v \mapsto t^{\Gamma \dominates (i,s)}_s(\psi_i)(\gamma\dominates[v/(i,s)])
\\
& = \; \lfp (v \mapsto \sem{\psi_i}_{\rho^{\gamma\dominates [v/(i,s)]}}(s)) & & \text{(by induction hypothesis)} \\
&  = \; \sem{X_i}_{\rho^{\gamma'}}(s)  & & \text{(by definition of $\rho^{\gamma'}$)} \\
&  = \; \sem{X_i}_{\rho^{\gamma}}(s)  & & \text{(because $\rho^{\gamma}$ and $\rho^{\gamma'}$ agree on $X_i$)}
\end{align*}

In the case that $\psi'$ is a least fixed point $\mu_i X_i. \, \psi_i$, we have
\begin{align*}
t^\Gamma_s (\mu  X_i. \, \psi_i) (\gamma) \; & = \;
  (\mu   x_{i,s} . \; t^{\Gamma \cup \{(i,s)\}}_s(\psi_i))(\gamma)
\\
& = \; \lfp (v \mapsto t^{\Gamma \cup \{(i,s)\}}_s(\psi_i)(\gamma[v/(i,s)]))
\\
& = \; \lfp (v \mapsto \sem{\psi_i}_{\rho^{\gamma \dominates [v/(i,s)]}}(s))
&& \text{(by induction hypothesis)}
\\
& = \; \sem{\mu X_i.\, \psi_i}_{\rho^{\gamma}}(s) \enspace .
\end{align*}
Here, the last equality follows from the definition of $\rho^\gamma$, since $\Gamma \cap (\{i\}\times S) = \emptyset$, due to the assumption relating $\Gamma$ and bound variables. 

The case in which $\psi'$ is a greatest fixed point is analogous.
\end{proof}

The main result of this section is merely a special case of the previous lemma, 
relativized to an arbitrary formula $\phi$ and finite PNTS.
\begin{proposition}
For any closed  $\Lukmu$ formula $\phi$, and state $s$ in a finite PNTS, it holds that $\sem{\phi}(s)\! =\! t^\emptyset_s(\phi)$.
\end{proposition}

\section{Related and future work}

The first encodings of probabilistic temporal logics in a probabilistic version of the modal $\mu$-calculus were given in~\cite{CPN99}, where a version $\PCTLstar$, tailored to processes exhibiting probabilistic but not nondeterministic choice,  was translated into a 
non-quantitative probabilisitic variant of the $\mu$-calculus, which included 
explicit (probabilistic) path quantifiers but disallowed fixed-point alternation. 

In their original paper on quantitative  $\mu$-calculi~\cite{HM96}, Huth and Kwiatkowska attempted a model checking algorithm for alternation-free formulas in the version of $\Lukmu$ with $\strongor$ and $\strongand$ but without $\weakand$, $\weakor$ and
scalar multiplication. Subsequently, 
several authors have addressed the problem of 
computing (sometimes approximating) 
fixed points for monotone functions combining
linear (sometimes polynomial) expressions with $\min$ and $\max$ operations;
see \cite{GS2011} for a summary. However, such work has focused on (efficiently) finding outermost (simultaneous) fixed-points for systems of equations whose underlying monotone functions are continuous. The nested fixed points considered in the present paper give rise to the complication of non-continuous functions.

It would be interesting to run experimental comparison of the iterative algorithm for calculating the value of a $\mu$-term against the reduction to  linear arithmetic. As mentioned in Section \ref{section:comparison}, we expect the direct algorithm to work better in practice than the non-elementary upper bound on its complexity given by our crude analysis, suggests. 
Furthermore, as a natural generalization of the approximation approach to computing fixed points, the direct algorithm should be amenable to optimizations such as the simultaneous solution of adjacent fixed points of the same kind, and the reuse of previous approximations when applicable due to  monotonicity considerations. 
It would also be interesting to obtain improved bounds on the complexity of calculating the value of a {\L}ukasiewicz $\mu$-term. Kyriakos Kalorkoti (private communication) has suggested one way of obtaining an exponential upper bound on the run-time. As regards lower bounds, it is easy to show that the problem is at least as hard as the \emph{simple stochastic game} problem \cite{Condon1992}. We wonder if computing the value of $\mu$-terms can also be shown to be in $\textrm{NP} \cap \textrm{co-NP}$.

Our results on $\Lukmu$ are a contribution towards the development of a robust theory of fixed-point probabilistic logics. The beginnings of a systematic study are carried out in \cite{MioThesis}, where an even richer logic is considered, 
including the operations of multiplication $( x \times y \!=\! x y)$ and co-multiplication ($x \otimes y\! =\! x + y -xy$). However, this extension has the disadvantage that simple formulas can have irrational (though algebraic) values. For example, consider the term $t= \mu x. \big(((x\otimes x)\times (x\otimes x))\sqcup \underline{\frac{1}{4}}\big)$. The value of this term is the least real number $x\!\in\![0,1]$ such that $x\geq \frac{1}{4}$ and $x\!=\! (2x -x^2)^2$, which is 
$\frac{3-\sqrt{5}}{2}$.  On the positive side, one can adapt the result of Proposition \ref{proposition:mu-term:arithmetic}, using an embedding into the first order theory of real-closed fields, to prove that, although irrational, the values of \L ukasiewicz terms extended with (co)multiplications are always algebraic. In a related direction, a fixed-point logic combining  \L ukasiewicz and (co)multiplication connectives, but based on the Brouwer fixed-point theorem (see discussion in Section \ref{section_lmu}), is investigated by Spada in \cite{Spada2008b}.

Overall, the {\L}ukasiewicz modal $\mu$-calculus of this paper seems to offer a good balance between expressivity (e.g., it can encode  $\PCTL$) and 
algorithmic  approachability. In addition, the first author has recently shown in \cite{MIO2014a} that the process equivalence for probabilistic nondeterministic transition systems characterised by $\Lukmu$ formulas is the standard notion of \emph{convex bisimilarity}~\cite{S95}. Thus the quantitative approach to probabilistic $\mu$-calculi may be considered equally suitable as a mechanism for 
characterising process equivalence as the  non-quantitative $\mu$-calculi 
advocated for this specific purpose in~\cite{CPN99} and~\cite{DvG2010}.

One further result of \cite{MioThesis} is an interpretation of the  {\L}ukasiewicz modal $\mu$-calculus, indeed its extension with (co)multiplication, in terms of a generalized notion of two-player stochastic game. However, although naturally interpreting the meaning of the connectives $\{\sqcup,\sqcap,\times,\otimes\}$, the game semantics given to  the connectives of  strong disjunction $\oplus$ and conjunction $\odot$ is indirect and intricate. An interesting direction for future research is to design an alternative game semantics that offers a more natural interpretation to  the connectives of the {\L}ukasiewicz $\mu$-calculus.

Another problem is to find sound and complete axiomatizations of the {\L}ukasiewicz modal $\mu$-calculus. Such an axiomatization could, for example, constitute a valuable tool for proving, or disproving, a long standing open problem in the literature: the decidability of the satisfiability property for \PCTL\ formulas. Axiomatizing the {\L}ukasiewicz modal $\mu$-calculus appears to be a difficult problem as similar results (e.g., axiomatization of Kozen's $\mu$-calculus) required significant efforts to be solved~\cite{Wal2000}.
A complete axiomatization of {\L}ukasiewicz fuzzy logic is presented in \cite{RMV2011}. It might serve as a basis for developing an axiomatization of  {\L}ukasiewicz $\mu$-calculus (i.e., without modalities) and, in another direction, of the fixed-point free fragment of {\L}ukasiewicz modal $\mu$-calculus. 

Further research will have to explore the relations between quantitative $\mu$-calculi such as $\Lukmu$ and other frameworks for verification and design of probabilistic systems. Important examples include the \emph{abstract probabilistic automata} of~\cite{DKLLPSW}, the compositional \emph{assume-guarantee} techniques of~\cite{KNPQ2010,FKNPQ2011} and the recent \emph{p-automata}  of~\cite{pautomata2012}. In particular, with respect to the latter formalism, we note that the acceptance condition of p-automata is specified in terms of stochastic games whose configurations may have preseeded threshold values whose action closely resembles that of the threshold modalities considered in this work (Definition \ref{thresholds_encoding}). Exploring the relations between p-automata games and $\Lukmu$-games~\cite{MioThesis} could be a fruitful topic for investigation.

\paragraph{Acknowledgements}
We thank Kousha Etessami, Kyriakos Kalorkoti, Enrico Marchioni, Grant Passmore, Phil Scott, Luca Spada, Colin Stirling and the anonymous reviewers for helpful comments  and for pointers to the literature.

\bibliography{biblio}

\begin{thebibliography}{10}

\bibitem{BaierKatoenBook}
C.~Baier and J.~P. Katoen.
\newblock {\em Principles of {M}odel {C}hecking}.
\newblock The {MIT} Press, 2008.

\bibitem{BA1995}
A.~Bianco and L.~de~Alfaro.
\newblock Model checking of probabilistic and nondeterministic systems.
\newblock In {\em Proceedings of Foundations of Software Technology and
  Theoretical Computer Science (FSTTCS)}, volume 1026 of {\em Lecture Notes in
  Computer Science}, pages 499--513. Springer Berlin Heidelberg, 1995.

\bibitem{BSS1989}
L.~Blum, M.~Shub, and S.~Smale.
\newblock On a theory of computation and complexity over the real numbers:
  {NP}-completeness, recursive functions and universal machines.
\newblock {\em Bulletin of the {AMS}}, 21(1), 1989.

\bibitem{BJW2005}
B.~Boigelot, S.~Jodogne, and P.~Wolper.
\newblock An effective decision procedure for linear arithmetic over the
  integers and reals.
\newblock {\em ACM Transactions on Computational Logic}, 6(13):614--633, 2005.

\bibitem{CK2013}
P.~Cintula and T.~Kroupa.
\newblock Simple games in {\l}ukasiewicz calculus and their cores.
\newblock {\em Kybernetika}, 3:404--419, 2013.

\bibitem{CPN99}
R.~Cleaveland, S.~P. Iyer, and M.~Narasimha.
\newblock Probabilistic temporal logics via the modal mu-calculus.
\newblock In {\em Proceedings of the Second International Conference on
  Foundations of Software Science and Computation Structures (FoSSaCS)}, pages
  288 -- 305. Springer-Verlag London, 1999.

\bibitem{Condon1992}
A.~Condon.
\newblock The {C}omplexity of {S}tochastic {G}ames.
\newblock {\em Information and Computation}, 96:203--224, 1992.

\bibitem{AM04}
L.~{de Alfaro} and R.~Majumdar.
\newblock Quantitative solution of omega-regular games.
\newblock {\em Journal of Computer and System Sciences}, 68:374 -- 397, 2004.

\bibitem{DKLLPSW}
B.~Delahaye, J.~P. Katoen, K.~Larsen, A.~Legay, M.~Pedersen, F.~Sher, and
  A.~Wasowski.
\newblock Abstract probabilistic automata.
\newblock {\em Information and Computation}, 232:66--106, 2013.

\bibitem{DvG2010}
Y.~Deng and R.~van Glabbeek.
\newblock Characterising probabilistic processes logically.
\newblock In {\em Proceedings of Logic for Programming, Artificial
  Intelligence, and Reasoning (LPAR)}, volume 6397 of {\em Lecture Notes in
  Computer Science}, pages 278--293. Springer, 2010.

\bibitem{Ferrante1975}
J.~Ferrante and C.~Rackoff.
\newblock A decision procedure for the first order theory of real addition with
  order.
\newblock {\em {SIAM} Journal of Computing}, 4(1):69--76, 1975.

\bibitem{FKNPQ2011}
V.~Forejt, M.~Kwiatkowska, G.~Norman, D.~Parker, and H.~Qu.
\newblock Quantitative multi-objective verification for probabilistic systems.
\newblock In {\em Proceedings of Tools and Algorithms for the Construction and
  Analysis of Systems (TACAS)}, volume 6605 of {\em Lecture Notes in Computer
  Science}, pages 112--127. Springer, 2011.

\bibitem{GS2011}
T.~M. Galwitza and H.~Seidl.
\newblock Solving systems of rational equations through strategy iteration.
\newblock {\em {ACM} Transactions on Programming Languages and Systems}, 33(3),
  2011.

\bibitem{Gerla2001}
B.~Gerla.
\newblock Rational {\l}ukasiewicz logic and {DMV}-algebras.
\newblock {\em Neural Networks World}, pages 579--584, 2001.

\bibitem{Gerla2001b}
B.~Gerla.
\newblock Rational {\l}ukasiewicz logic and {DMV}-algebras.
\newblock arXiv:1211.5485, Nov 2012.

\bibitem{HJ98}
P.~H\'{a}jek.
\newblock {\em Metamathematics of Fuzzy Logic}.
\newblock Springer, 2001.

\bibitem{HM96}
M.~Huth and M.~Kwiatkowska.
\newblock Quantitative analysis and model checking.
\newblock In {\em Proceedings of IEEE Symposium on Logic in Computer Science
  (LICS)}, pages 111--122. IEEE, 1997.

\bibitem{pautomata2012}
M.~Huth, N.~Piterman, and D.~Wagner.
\newblock p-{A}utomata: {N}ew {F}oundations for discrete-time {P}robabilistic
  {V}erification.
\newblock {\em Performance Evaluation}, 69(7), 2012.

\bibitem{JW96}
D.~Janin and I.~Walukiewicz.
\newblock On the expressive completeness of the propositional mu-calculus with
  respect to monadic second order logic.
\newblock {\em Lecture Notes in Computer Science}, 1119:263--277, 1996.

\bibitem{Kozen83}
D.~Kozen.
\newblock Results on the propositional mu-calculus.
\newblock {\em Theoretical Computer Science}, 27(3):333--354, 1983.

\bibitem{KNPQ2010}
M.~Kwiatkowska, G.~Norman, D.~Parker, and H.~Qu.
\newblock Assume-guarantee verification for probabilistic systems.
\newblock In {\em Proceedings of Tools and Algorithms for the Construction and
  Analysis of Systems (TACAS)}, volume 6015 of {\em Lecture Notes in Computer
  Science}, pages 23--37. Springer, 2010.

\bibitem{Mader1995}
A.~Mader.
\newblock Modal {$\mu$}-calculus, model checking and gau{\ss} elimination.
\newblock In {\em Proceedings of Tools and Algorithms for the Construction and
  Analysis of Systems (TACAS)}, volume 1217 of {\em Lecture Notes in Computer
  Science}, pages 72--88. Elsevier, 1995.

\bibitem{MM07}
A.~McIver and C.~Morgan.
\newblock Results on the quantitative {$\mu$}-calculus q{M}{$\mu$}.
\newblock {\em ACM Transactions on Computational Logic}, 8(1), 2007.

\bibitem{McNaughton1951}
R.~Mc{N}aughton.
\newblock A theorem about infinite-valued sentential logic.
\newblock {\em The Journal of Symbolic Logic}, 16:1--13, 1951.

\bibitem{MioThesis}
M.~Mio.
\newblock {\em Game Semantics for Probabilistic {$\mu$}-Calculi}.
\newblock PhD thesis, School of Informatics, University of Edinburgh, 2012.

\bibitem{MIO2012a}
M.~Mio.
\newblock On the equivalence of denotational and game semantics for the
  probabilistic {$\mu$}-calculus.
\newblock {\em Logical Methods in Computer Science}, 8(2), 2012.

\bibitem{MIO2012b}
M.~Mio.
\newblock {P}robabilistic {M}odal {$\mu$}-{C}alculus with {I}ndependent
  product.
\newblock {\em Logical Methods in Computer Science}, 8(4), 2012.

\bibitem{MIO2014a}
M.~Mio.
\newblock Upper-{E}xpectation bisimilarity and {\l}ukasiewicz
  {$\mu$}-{C}alculus.
\newblock In {\em Proceedings of Foundations of Software Science and
  Computation Structures (FoSSaCS)}, volume 8412 of {\em Lecture Notes in
  Computer Science}, pages 335--350. Elsevier, 2014.

\bibitem{MM97}
C.~Morgan and A.~McIver.
\newblock A probabilistic temporal calculus based on expectations.
\newblock In {\em Proceedings of Formal Methods Pacific}, pages 4--22.
  Springer-Verlag, 1997.

\bibitem{MundiciBook}
D.~Mundici.
\newblock {\em Advanced {\L}ukasiewicz Calculus and MV-Algebras}.
\newblock Trends in Logic. Springer-Verlag, 2011.

\bibitem{RMV2011}
A.~D. Nola and I.~Leustean.
\newblock Riesz {MV}-algebras and their logic.
\newblock In {\em Proceeding of {C}onference of the {E}uropean {S}ociety for
  {F}uzzy {L}ogic and {T}echnology (EUSFLAT)}, pages 140--145. Atlantis Press,
  2011.

\bibitem{S95}
R.~Segala.
\newblock {\em Modeling and Verification of Randomized Distributed Real-Time
  Systems}.
\newblock PhD thesis, MIT, 1995.

\bibitem{Spada2008b}
L.~Spada.
\newblock {\L$\Pi$} logic with fixed points.
\newblock {\em Archive for Mathematical Logic}, 47:1260--1267, 2008.

\bibitem{Spada2008}
L.~Spada.
\newblock {$\mu$}{MV}-algebras: {A}n approach to fixed points in {\l}ukasiewicz
  logic.
\newblock {\em Fuzzy Sets and Systems}, 159(10):1260--1267, 2008.

\bibitem{Stirling96}
C.~Stirling.
\newblock {\em Modal and temporal logics for processes}.
\newblock Springer, 2001.

\bibitem{Wal2000}
I.~Walukiewicz.
\newblock Completeness of {K}ozen's axiomatisation of the propositional
  {$\mu$}-calculus.
\newblock {\em Information and Computation}, 157, 2000.

\end{thebibliography}
\bibliographystyle{abbrv}

\end{document}